\documentclass[11pt,a4paper]{article}
\usepackage[utf8]{inputenc}
\usepackage[T1]{fontenc}
\usepackage{mathptmx}
\usepackage[dvipsnames]{xcolor}
\usepackage{microtype}
\usepackage[margin=2.5cm]{geometry}
\usepackage{amsmath,amssymb,amsthm,mathtools}
\usepackage{enumitem}
\usepackage{tikz}
\usepackage{bm}
\usepackage{caption}
\usetikzlibrary{arrows.meta,decorations.pathreplacing,calc}
\usepackage[colorlinks=true,linkcolor=MidnightBlue,citecolor=MidnightBlue,urlcolor=MidnightBlue]{hyperref}

\theoremstyle{plain}
\newtheorem{theorem}{Theorem}

\newtheorem{corollary}{Corollary}
\theoremstyle{definition}
\newtheorem{definition}{Definition}
\theoremstyle{remark}

\newcommand{\PP}{\mathcal{P}}
\newcommand{\U}{\mathcal{U}}
\newcommand{\bb}[1]{\mathbf{#1}}
\newcommand{\rr}{\bb{r}}
\newcommand{\nn}{\bb{n}}
\newcommand{\rpar}{\rr^{\parallel}}
\newcommand{\xr}{x_{\rr}}
\newcommand{\simp}{\triangle_1}

\newcommand{\Var}{\operatorname{Var}}
\newcommand{\Cov}{\operatorname{Cov}}
\newcommand{\avg}[1]{\langle #1\rangle}
\newcommand{\avgU}[1]{\langle #1\rangle_{\U}}
\newcommand{\rhobar}{\bar\rho}
\newcommand{\ind}{\chi}

\title{\textbf{Are quantum measurements universal?}}
\author{Massimiliano Sassoli de Bianchi\\[6pt]
\small Center Leo Apostel for Interdisciplinary Studies, Brussels Free University, Brussels, Belgium\\[4pt]
\small\textit{27 October 2026}}
\date{}

\begin{document}
\maketitle

\begin{abstract}
\noindent In 2014, together with D. Aerts, we showed -- within the extended Bloch representation of quantum mechanics -- that the Born rule can be understood as resulting from a universal average over all the possible ways in which the measurement-interactions can be selected, a result that was obtained by averaging over discretized (cellular) probability densities and then taking the continuum limit. The construction has since been questioned, in the context of Bertrand's paradox, by N.~Shackel, who observed that the result depends on the specific discretization adopted, that the sequence of averages is constant, and that the sense in which the average is taken over `all' probability densities remains unclear. In this article, we take these criticisms as a stimulus to re-examine the notion of universal measurement. We show that the key to the whole matter is to recognize that what makes a random way of selecting a measurement `universal' is that it does not privilege any portion of the abstract structure describing the available interactions, a property we call indifference, which we adopt as the defining one. We prove that every universal measurement, however it is constructed, produces exactly the Born probabilities, so that the universal average is unique, and the cellular construction of 2014 is one among infinitely many ways of obtaining it. For cellular measurements, the result already holds exactly at every finite resolution, with no limit to be taken, which explains the constancy of the sequence of averages and dispenses with the combinatorial identities used in the original proof. We then show that non-trivial universal measurements exist, and that they must involve deterministic measurements: the random mixtures of deterministic measurements known in probability theory as Dirichlet processes are universal measurements whose support contains every possible way of selecting an interaction, whereas a universal measurement free of deterministic components necessarily selects the uniform density, so that for it the Born rule holds not only on average but with certainty. The deviations from the quantum statistics that deterministic components can produce in sequences of measurements are governed by a single parameter, which is in principle accessible to experiment. The only geometric input of the whole construction is the Euclidean length (or volume) of the measurement membrane, together with the orthogonal landing of the point-particle representing the state, which is where the content of Gleason's theorem lies. Finally, we discuss the physical meaning of these results, indicate how they generalize to measurements with an arbitrary number of outcomes, and comment on their bearing on Bertrand’s paradox.
\end{abstract}

\section{Introduction}\label{sec:intro}

The idea that quantum measurements might be \emph{universal measurements} has a long history, which it is useful to recall briefly, as it explains both the motivation of the present article and the nature of the questions it addresses.

The starting point is the \emph{hidden-measurement} approach to quantum probabilities, proposed by Diederik Aerts in the eighties \cite{Aerts1986}. In this approach, the indeterminism of a quantum measurement is not attributed to a lack of knowledge about the state of the measured entity (which would lead to a classical, Kolmogorovian probability model), but to a lack of knowledge about the \emph{interaction} between the measuring apparatus and the entity, which is each time selected in an uncontrollable way. The simplest illustration is the so-called \emph{spin quantum machine}, or \emph{sphere model}: a point particle (descriptive of the system's state) on the surface of a sphere, an abstract elastic band stretched between two opposite points of the sphere, the particle falling orthogonally onto the band and sticking to it, then the band disintegrating at some unpredictable point, and contracting, thus drawing the particle to one of its two end points. When the band is uniformly disintegrable, one obtains exactly the probabilities of a spin-$\frac12$ measurement. When it is not, one obtains probability models which are neither classical nor quantum, the so-called \emph{intermediate} models, of which the $\epsilon$-model is the paradigmatic example \cite{Aerts1993, Aerts1994, AertsDurt1994, Aerts1995}.

The existence of this whole spectrum of possible measurements immediately raised a question: why, in the microscopic world, would it be precisely the (uniform) quantum measurements that are realized? The answer that was conjectured in \cite{Aerts1998,Aerts1999} was that quantum measurements could be \emph{universal} measurements, i.e., the result of a huge average over all the possible ways of selecting an interaction, so that the quantum probabilities would be those of a ``first-order non-classical theory,'' in the same way as the first term of a Taylor expansion is a first-order approximation to any smooth function. A heuristic argument in favor of this conjecture was based on Gleason's theorem \cite{Gleason1957,Sassoli2026}: if the transition probability of a universal measurement depends only on the state before the measurement and the state after it, then it must be the quantum transition probability. But the conjecture could not be given, at the time, a precise mathematical content, because any attempt to define a uniform average over ``all possible ways of selecting an interaction'' was immediately confronted with the difficulties of Bertrand's paradox \cite{Bertrand1889}.

The scientific collaboration between D. Aerts and the author started in 2013 precisely on this issue, initially from a position of disagreement, as the author was then exploring a different explanation of the emergence of quantum measurements, based on the notion of \emph{robustness} \cite{Raedt2014,AS2015b}, and was rather skeptical about the possibility of giving a precise meaning to the notion of universal measurement. The construction of a simple toy model, however, confirmed the conjecture instead of refuting it, and led to the general result published in \cite{AS2014a} (see also \cite{AS2015a,ASbook}): in the \emph{extended Bloch representation} of quantum mechanics, where a measurement with $N$ outcomes is represented by an abstract $(N-1)$-dimensional elastic and disintegrable membrane, the Born rule corresponds to a uniform membrane, and its uniform disintegrability can in turn be understood as resulting from a \emph{universal average} over all the possible non-uniform ways in which the membrane can disintegrate. The average was defined by first discretizing the membrane into $n$ cells of equal size, considering all the $2^n-1$ ``cellular'' probability densities taking a constant value on a subset of the cells and vanishing elsewhere, averaging the outcome probabilities uniformly over them, and finally taking the limit $n\to\infty$, in the spirit of the definition of the Wiener measure. The same construction was then also applied to Bertrand's paradox, providing what we called a solution of its \emph{hard problem} \cite{AS2014b}.

It is in this latter context that our construction was examined, with great care, by N. Shackel in his recent monograph on Bertrand's paradox and the principle of indifference \cite{Shackel2024}, where a whole chapter is devoted to our universal average. Shackel's observations can be summarized as follows. First, the argument by which we showed that every probability density can be approximated by a sequence of cellular densities proves less than it claims, because the approximating sequence is allowed to depend on the interval under consideration (what Shackel calls an ambiguity in the order of the quantifiers). Second, the three constraints on which the construction rests (cells of equal length, two values only, uniform weights) are not justified, and as soon as more than two values are allowed there are uncountably many cellular densities, so that a uniform average over them seems no longer definable. Third, the sequence of averages over the cellular densities is in fact constant, in the sense that each of its terms is already the uniform probability, so that the limit $n\to\infty$ adds nothing, and the theorem looks ``trivial.'' Fourth, in the context of Bertrand's paradox, the result depends on the parameterization adopted (what Shackel calls a ``naming''), so that the paradox would recur at the level of the choice of the parameterization.

We consider these observations as valuable, and the present article, which is to be regarded as the natural continuation of the effort initiated in \cite{AS2014a}, has been stimulated by them. More precisely, we shall see that they can all receive a clear answer as soon as one recognizes that a universal average is an \emph{invariant average}, i.e., that the ``uniform average over all possible probability densities'' of \cite{AS2014a} is best understood as an average which does not privilege any portion of the elastic band, for the two-outcome situation, or of more general simplicial structures, for the $N$-outcome situation. Once this is recognized, the following becomes clear.

\begin{enumerate}[leftmargin=*]
\item The universal average is \emph{unique}: any random way of selecting a $\rho$-measurement (a measurement characterized by a probability density $\rho$) which is invariant in the above sense produces exactly the uniform density $\rho_u$, hence the Born rule. The cellular construction presented in \cite{AS2014a} is one way, among infinitely many, of obtaining this uniform average through an indifferent selection of $\rho$-measurements. Its use of two density values and uniform weights is not essential; what matters is the invariance of the selection law under permutations of equal-size cells (Section~\ref{sec:indifference}).
\item The result holds \emph{exactly} at every finite resolution: the average over any invariant family of cellular densities with $n$ cells is the uniform density, for every $n$, with no limit to be taken, and this explains why the sequence of averages in \cite{AS2014a} is constant, and the origin of the combinatorial identities that were used in the proof by induction (Section~\ref{sec:finite} and Appendix~\ref{app:comb}).
\item The ``huge average over all possible ways of selecting an interaction'' can be given a precise meaning, provided one admits, as we did in \cite{AS2014a}, \emph{deterministic} measurements, in which the band disintegrates at a predetermined point. Dirichlet processes, which are random mixtures of deterministic measurements, provide non-trivial universal measurements whose weak support contains every probability distribution on the band (Section~\ref{sec:pure}).
\item A universal measurement which almost surely selects no deterministic components necessarily selects the uniform density: the Born rule then governs the outcome probabilities of each individual run, not merely their average. More generally, any non-trivial universal measurement must admit deterministic components, whose contribution to fluctuations and correlations in sequences of measurements is quantified by the coincidence probability $\beta$ (Sections~\ref{sec:rigidity}--\ref{sec:sequences}).
\item The definition of indifference requires a notion of size, here provided by the Euclidean geometry of the extended Bloch sphere: no such notion can be extracted from the set of all densities alone (Section~\ref{sec:length}).
\end{enumerate}

\noindent The article is organized as follows. In Section~\ref{sec:bloch}, we briefly recall the two-outcome measurement process in the extended Bloch representation, so as to arrive at the key expressions for the outcome probabilities of a general $\rho$-measurement. In Section~\ref{sec:2014}, we recall the universal average as defined in \cite{AS2014a}, and we formulate the questions it leaves open. In Section~\ref{sec:random}, we introduce the notion of a $\U$-measurement, i.e., of a random $\rho$-measurement, and we observe a simple but decisive fact: the outcome statistics of a $\U$-measurement are those of an ordinary $\rho$-measurement, with $\rho$ the \emph{average} density. In Section~\ref{sec:indifference}, we then single out, among the $\U$-measurements, those that deserve to be called universal. Sections~\ref{sec:indifference}--\ref{sec:length} contain the results listed above and, in Section~\ref{sec:general}, we indicate how everything generalizes to measurements with $N$ outcomes, and we briefly comment on Bertrand's paradox. Section~\ref{sec:conclusion} contains our conclusions, and some additional technical proofs are deferred to the appendices.

\section{Two-outcome measurements}\label{sec:bloch}

We recall the essential elements of the model of \cite{AS2014a}, in the two-outcome case, referring to that article for all details. A two-dimensional quantum entity (for instance a spin-$\frac12$ entity) is described, in the Bloch representation, by a real vector $\rr$ in the three-dimensional unit ball: a unit vector for a vector-state, a shorter vector for an operator-state (density matrix). A measurement with two outcomes $a_1$ and $a_2$, associated with two orthogonal rank-$1$ projection operators $P_{a_1}$ and $P_{a_2}$, is represented by two opposite unit vectors $\nn_1$ and $\nn_2=-\nn_1$, and by the diameter $\simp$ of the ball joining them, which one can think of as an abstract elastic and disintegrable (i.e., breakable) band of length $2$, stretched between $\nn_1$ and $\nn_2$.

The measurement process unfolds in two stages. In the first, deterministic, stage, the point-particle representing the state is attracted by the band and ``falls'' orthogonally onto it, reaching the position
\begin{equation}\label{eq:rpar}
\rpar=(\rr\cdot\nn_1)\,\nn_1
\end{equation}
It is convenient to introduce a coordinate $x\in[-1,1]$ along the band, writing a generic point of $\simp$ as $\bm{\lambda}=x\,\nn_1$, so that $x=-1$ corresponds to $\nn_2$ and $x=1$ to $\nn_1$. The particle then sits at the coordinate
\begin{equation}\label{eq:xr}
\xr\equiv\rr\cdot\nn_1=\|\rr\|\cos\theta(\rr,\nn_1)
\end{equation}
which for a vector-state is simply $\cos\theta(\rr,\nn_1)$, $\theta(\rr,\nn_1)$ being the angle between $\rr$ and $\nn_1$. In the notation of \cite{AS2014a}, $\xr=r_1^\parallel-r_2^\parallel$, where $\rpar=r_1^\parallel\nn_1+r_2^\parallel\nn_2$ is the decomposition of $\rpar$ in the convex combination of the two end points, with $r_1^\parallel+r_2^\parallel=1$. The presence of the particle divides the band into two regions, $A_1=[-1,\xr]$, between $\nn_2$ and $\rpar$, and $A_2=[\xr,1]$, between $\rpar$ and $\nn_1$ (see Fig.~\ref{fig:band}).

In the second, indeterministic, stage, the band, which is unstable, disintegrates at some point $\bm{\lambda}=x\, \nn_1$ and its two pieces contract toward their respective anchor points. If the disintegrating point lies in $A_1$, the particle, attached to the piece anchored at $\nn_1$, is drawn to $\nn_1$, so that the outcome is $a_1$. If it lies in $A_2$, the particle is drawn to $\nn_2$ and the outcome is $a_2$. The point $\bm{\lambda}$ is the \emph{hidden variable}  of the measurement. Note, however, that it is not a hidden variable of the entity, but of the measurement interaction, selected each time in an unpredictable way. This is why the hidden-measurement approach is not affected by the no-go theorems \cite{Bell1966,KochenSpecker1967,Gudder1970}, which refer solely to hidden variables associated with the state of the measured entity. 
\begin{figure}[t]
\centering
\begin{tikzpicture}[scale=1.0]
\draw[->] (-4.2,1.6) -- (4.4,1.6) node[right] {$x$};
\draw[->] (-4,1.6) -- (-4,3.6) node[above] {$\rho(x)$};
\draw[thick,MidnightBlue,domain=-4:4,samples=80] plot (\x,{1.6+0.9*exp(-0.5*(\x+1.4)*(\x+1.4))+0.5*exp(-1.5*(\x-2.3)*(\x-2.3))+0.15});
\draw[dashed] (1.2,1.6) -- (1.2,3.4);
\node at (1.2,1.35) {$\xr$};
\node at (-4,1.35) {$-1$};
\node at (4,1.35) {$1$};
\draw[dashed,gray] (0,-0.6) ellipse (4 and 1);
\draw[line width=3pt,gray] (-4,-0.6) -- (4,-0.6);
\fill (-4,-0.6) circle (2.2pt) node[below left=1pt] {$\nn_2$};
\fill (4,-0.6) circle (2.2pt) node[below right=1pt] {$\nn_1$};
\fill (1.2,-0.6) circle (2.2pt) node[below=3pt] {$\rpar$};
\fill (1.2,0.354) circle (2.2pt) node[above right=1pt] {$\rr$};
\draw[->,>=Latex] (1.2,0.354) -- (1.2,-0.5);
\draw[decorate,decoration={brace,mirror,amplitude=5pt}] (-4,-1.85) -- (1.1,-1.85) node[midway,below=6pt] {$A_1$ (outcome $a_1$)};
\draw[decorate,decoration={brace,mirror,amplitude=5pt}] (1.3,-1.85) -- (4,-1.85) node[midway,below=6pt] {$A_2$ (outcome $a_2$)};
\end{tikzpicture}
\caption{The two-outcome measurement in the extended Bloch representation. The dashed circle is the great circle of the Bloch sphere containing $\nn_1$, $\nn_2$ and the vector-state $\rr$ (drawn in perspective), of which the elastic band $\simp$ is a diameter. The point-particle falls orthogonally from $\rr$ onto the band, reaching $\rpar$, of coordinate $\xr$. The band then breaks at a point selected according to the probability density $\rho(x)$: if the disintegrating point lies in $A_1$, the particle is drawn to $\nn_1$ (outcome $a_1$), otherwise to $\nn_2$ (outcome $a_2$).}
\label{fig:band}
\end{figure}

If the band is made of a uniform substance, all its points being equally likely to disintegrate, the probability of the outcome $a_1$ is the ratio between the length $1+\xr$ of $A_1$ and the total length $2$ of the band, and similarly for $a_2$:
\begin{equation}\label{eq:born}
\PP(\rr\to\nn_1)=\frac{1+\xr}{2}=\frac{1+\rr\cdot\nn_1}{2},\quad \PP(\rr\to\nn_2)=\frac{1-\xr}{2}=\frac{1+\rr\cdot\nn_2}{2}
\end{equation}
which is exactly the Born rule for a two-dimensional entity. This is the content of Section~7 of \cite{AS2014a}.

If the band is not uniform, its disintegrability being described by a probability density $\rho(x)$, $\rho\ge0$, $\int_{-1}^1\rho(x)dx=1$, the outcome probabilities become (Eq.~(80) of \cite{AS2014a}):
\begin{equation}\label{eq:rho-prob}
\PP(\rr\to\nn_1|\rho)=\int_{-1}^{\xr}\rho(x)\,dx,\quad \PP(\rr\to\nn_2|\rho)=\int_{\xr}^{1}\rho(x)\,dx
\end{equation}
We shall call such a measurement a \emph{$\rho$-measurement}. The uniform density $\rho_u(x)=1/2$ gives back \eqref{eq:born}. Hence, a \emph{$\rho_u$-measurement} is a quantum measurement. It is useful to introduce the \emph{cumulative function} of $\rho$,
\begin{equation}\label{eq:F}
F_\rho(x)\equiv\int_{-1}^{x}\rho(y)\,dy,\quad F_\rho(1)=1
\end{equation}
which is non-decreasing (with $F_\rho(-1)=0$ unless $\rho$ has a delta component at $x=-1$, in which case $F_\rho(-1)$ equals its weight), and in terms of which \eqref{eq:rho-prob} becomes
\begin{equation}\label{eq:F-prob}
\PP(\rr\to\nn_1|\rho)=F_\rho(\xr),\quad \PP(\rr\to\nn_2|\rho)=1-F_\rho(\xr)
\end{equation}
In other words, the probability model of a two-outcome $\rho$-measurement, i.e., the transition probability as a function of the initial state, is nothing but the cumulative function of $\rho$, and the quantum probability model corresponds to the linear function $F_u(x)=(1+x)/2$. We will make repeated use of this simple observation.

As in \cite{AS2014a}, we allow $\rho$ to contain Dirac delta components. The density $\rho(x)=\delta(x-x_0)$ describes a band that can only disintegrate at the predetermined point $x_0$, which corresponds to a \emph{deterministic} measurement. Indeed, given the state, the outcome is predetermined: $a_1$ if $\xr>x_0$ and $a_2$ if $\xr<x_0$. The density $c_1\delta(x+1)+c_2\delta(x-1)$ describes a so-called \emph{solipsistic} measurement, whose outcome probabilities ($c_1$ and $c_2$) do not depend on the state. The density $\rho_\epsilon(x)=\frac{1}{2\epsilon}\ind_{[-\epsilon,\epsilon]}(x)$ describes the $\epsilon$-model, which interpolates between a deterministic measurement ($\epsilon\to0$) and the quantum one ($\epsilon\to1$) and which, for $0<\epsilon<1$, is neither classical nor quantum \cite{Aerts1993, Aerts1994, AertsDurt1994, Aerts1995}. 

Note that, mathematically, a density with delta components is just a probability measure on $[-1,1]$. When $\rho$ has delta components, the integrals over the regions $A_1$ and $A_2$ in \eqref{eq:rho-prob} are to be understood with the point $\xr$ assigned to one of the two regions, say $A_1=[-1,\xr]$ and $A_2=(\xr,1]$; which one is irrelevant for all the averages considered below, since a delta component sitting exactly at $\xr$ corresponds to the exceptional situation of unstable equilibrium discussed in \cite{AS2014a}.

\section{Averaging over cellular densities}\label{sec:2014}

The question addressed in Section~13 of \cite{AS2014a} was the following. The Born rule \eqref{eq:born} is obtained when the band is uniformly disintegrable. But why would the band be uniform? The hidden-measurement interpretation offers a natural answer: a quantum measurement is a situation of \emph{maximal lack of knowledge}, not only about which interaction $\bm{\lambda}$ is actualized, but also about \emph{the way} the interaction is actualized, i.e., about the density $\rho$. If nothing is known about $\rho$, all the $\rho$ should be considered as equally possible, and the measurement that is actually performed should be described by an average over all of them. To put it another way, if nothing in an experiment allows one to control the way an outcome is chosen, all possible ways (all possible $\rho$) are a priori possible and must be taken into account when calculating probabilities.  We called such a measurement a \emph{universal measurement}, and we conjectured, and then proved in a specific sense, that the universal average of the outcome probabilities \eqref{eq:rho-prob} is exactly the Born rule.

The definition adopted in \cite{AS2014a}  was the following. The interval $[-1,1]$ is partitioned into $n$ \emph{elementary cells} of equal length $2/n$, and one considers the \emph{cellular} probability densities
\begin{equation}\label{eq:cellular}
\rho_n(x)=\frac{n}{2\, n^d}\,\chi_n(x),
\end{equation}
where $\chi_n(x)$ is a step function taking the value $1$ in the $n^d$ disintegrable cells, $n^d\ge1$, and $0$ in the remaining $n-n^d$ non-disintegrable cells (see Fig.~9 of \cite{AS2014a}; note that $d$ is just a superscript, not an exponent, indicating that $n^d$ refers to the disintegrable cells). Such densities describe bands which can only break in the disintegrable cells, uniformly within them. For a given $n$, there are exactly $2^n-1$ different cellular densities, and one defines the \emph{average probability} at level $n$ by giving the same weight to each of them:
\begin{equation}\label{eq:avg2014}
\avg{\PP(\rr\to\nn_i|\rho_n)}\equiv\frac{1}{2^n-1}\sum_{\rho_n}\PP(\rr\to\nn_i|\rho_n),\quad i=1,2
\end{equation}
where $\PP(\rr\to\nn_i|\rho_n)=\PP(x\in A_i|\rho_n)=\int_{A_i}\rho_n(x)dx$. The universal average is then defined as the infinite-cell limit
\begin{equation}\label{eq:univ2014}
\avg{\PP(\rr\to\nn_i)}^{\rm univ}\equiv\lim_{n\to\infty}\avg{\PP(\rr\to\nn_i|\rho_n)},\quad i=1,2
\end{equation}
and the theorem proved in \cite{AS2014a} states that
\begin{equation}\label{eq:thm2014}
\avg{\PP(\rr\to\nn_i)}^{\rm univ}=\PP(\rr\to\nn_i|\rho_u)=\frac{1+\rr\cdot\nn_i}{2},\quad i=1,2
\end{equation}
i.e., that the universal average reproduces the Born rule. 

The proof had two steps. In the first (Section~13.1 of \cite{AS2014a}), it was shown that cellular densities are \emph{dense} among all densities, in the sense that for every $\rho$ one can find cellular densities $\rho_n$ such that $\PP(\rr\to\nn_i|\rho_n)\to\PP(\rr\to\nn_i|\rho)$ as $n\to\infty$. The purpose of this step was to make sure that the average \eqref{eq:univ2014} could legitimately be considered as an average over \emph{all} densities. In the second step (Section~13.2), it was shown, by induction and using two combinatorial identities (Eqs.~(108) and (113) of \cite{AS2014a}), that for intervals $A_i$ made of whole cells
\begin{equation}\label{eq:avg-n-uniform}
\avg{\PP(\rr\to\nn_i|\rho_n)}=\PP(\rr\to\nn_i|\rho_{u;n}),\quad i=1,2
\end{equation}
where $\rho_{u;n}$ is the cellular density whose $n$ cells are all disintegrable, i.e., the uniform density, from which \eqref{eq:thm2014} follows in the limit.

This construction raises a number of questions, which we now formulate in our own words (they are essentially those raised by Shackel \cite{Shackel2024}).

\begin{enumerate}[label=(Q\arabic*),leftmargin=*]
\item \emph{Why these particular cellular densities?} The construction uses cells of equal length, densities taking only two values (one of which is zero), and equal weights for all of them. If a different discretization were used, say cells of different lengths, or densities with three or more values, or unequal weights, would the result be the same? Note that as soon as three values are allowed there are uncountably many cellular densities with $n$ cells, so that ``giving the same weight to each of them'' no longer corresponds to a finite average.
\item \emph{What does the limit add?} Equation \eqref{eq:avg-n-uniform} holds for every $n$ and the sequence of averages is constant. In other words, the limit $n\to\infty$ in \eqref{eq:univ2014} is the limit of a constant sequence. In what sense, then, is the result ``universal,'' and not just ``trivial''?
\item \emph{Is the average over `all' densities?} The density argument of the first step shows that every $\rho$ can be approximated by cellular densities, but this does not, by itself, make an average over the cellular densities of level $n$, followed by a limit, an average over all densities. Moreover, as Shackel pointed out, the approximation argument of \cite{AS2014a} was written with an ambiguous order of quantifiers (the approximating sequence was allowed to depend on the interval considered).
\item \emph{What is the role of the length?} The density $\rho_u$ is uniform with respect to the \emph{length} along the band, and equal cells are cells of equal length. If the band were parameterized differently, by a coordinate $y=h(x)$, the uniform density in $x$ would not be uniform in $y$. Where does the length come from, and is it an arbitrary choice?
\end{enumerate}

In the following sections we answer these questions. The answer to all of them rests on a single idea: the universal average is an \emph{invariant} average.

\section{Random $\rho$-measurements}\label{sec:random}

As we have explained, a $\rho$-measurement is a measurement in which a potential interaction $\bm{\lambda}$ is randomly actualized through a process weighted by the density $\rho$. Consider now a measurement in which the density $\rho$ is itself actualized at random, each time the measurement is performed, among all the possible (potential) densities. To describe such a measurement, we must specify the probability law according to which $\rho$ is selected. Let $\U$ denote this law. We shall then say that a $\U$-measurement is a measurement characterized by $\U$, in the same way as a $\rho$-measurement is a measurement characterized by $\rho$. In other words, a $\U$-measurement is a \emph{random $\rho$-measurement}, or, as we also said in \cite{AS2014a}, a ``measurement of measurements.''  In itself, a $\U$-measurement carries no claim of universality: it is just a mixture of measurements, and if $\U$ happens to be concentrated on a single density, the $\U$-measurement is nothing but the ordinary $\rho$-measurement. What turns a $\U$-measurement into a \emph{universal} measurement, in the sense conjectured in \cite{Aerts1998,Aerts1999}, is an additional condition on $\U$, expressing the absence of any knowledge about which portion of the band the selected measurement will favor. We formulate it in Section~\ref{sec:indifference}.

For instance, the average \eqref{eq:avg2014} corresponds to the  $\U_n$-measurement in which $\rho$ is selected, with the same probability $\frac{1}{2^n-1}$, among the $2^n-1$ cellular densities with $n$ cells. The law $\U_n$ is a probability law on a finite set of densities, but in general $\U$ can be a probability law on an infinite set of densities, discrete or continuous. In other words, many different $\U$-measurements can a priori be defined, each characterized by a specific set of potential probability densities that can be randomly actualized, and by the probabilities with which they are actualized.

Given a $\U$-measurement, the probability of the outcome $a_i$ is obtained by averaging the probability \eqref{eq:rho-prob} over the random $\rho$, i.e., by taking the \emph{mean value} of $\PP(\rr\to\nn_i|\rho)$ with respect to the law $\U$, which we denote by
\begin{equation}\label{eq:univ-prob}
\avgU{\PP(\rr\to\nn_i|\rho)}\equiv\int\PP(\rr\to\nn_i|\rho)\,d\U(\rho),\quad i=1,2
\end{equation}
For a finite law like $\U_n$ the integral is just the sum \eqref{eq:avg2014}. The following simple observation is the key to everything that follows. Since, according to (\ref{eq:rho-prob}), $\PP(\rr\to\nn_i|\rho)$ depends \emph{linearly} on $\rho$, the mean value of the probability is the probability of the mean density. So, if we define the \emph{average density} of the $\U$-measurement as
\begin{equation}\label{eq:avg-density}
\rhobar_\U\equiv\avgU{\rho}=\int\rho\,d\U(\rho)
\end{equation}
which is again a probability density on $[-1,1]$, then, exchanging the order of the two integrations, \eqref{eq:univ-prob} becomes
\begin{equation}\label{eq:key}
\avgU{\PP(\rr\to\nn_1|\rho)}=\int_{-1}^{\xr}\rhobar_\U(x)\,dx=\PP(\rr\to\nn_1|\rhobar_\U)
\end{equation}
and similarly for $\nn_2$. In other words, from the point of view of the outcome probabilities, a $\U$-measurement is an ordinary $\rho$-measurement, whose density is the average density, or, to put it more succinctly: a $\U$-measurement is a $\rhobar_\U$-measurement. 

The question in the title (are quantum measurements universal?) therefore becomes a question about average densities: \emph{for which laws $\U$ is the average density $\rhobar_\U$ the uniform density $\rho_u$?} The answer, as we are now going to show, is: for all those which are \emph{indifferent}, in a sense that we now define -- and it is precisely these that we shall call universal measurements.

\section{Indifference and universal measurements}\label{sec:indifference}

We consider the band $\simp$, which spans the interval $[-1,1]$, and divide it into $n$ \emph{cells} $\sigma^{(n)}_j$ of equal length $2/n$:
\begin{equation}\label{eq:cells}
\sigma^{(n)}_j=\Big[-1+(j-1)\frac2n,\ -1+j\frac2n\Big),\quad j=1,\dots,n
\end{equation}
the last cell being closed at $x=1$. For a density $\rho$, let
\begin{equation}\label{eq:masses}
m_j(\rho)\equiv\int_{\sigma^{(n)}_j}\rho(x)\,dx,\quad j=1,\dots,n,\quad \sum_{j=1}^n m_j(\rho)=1
\end{equation}
be the \emph{cell masses} of $\rho$, giving the probabilities that the band disintegrates in the different cells. Now imagine that the band, made of the non-uniform substance described by $\rho$, is cut along the cell boundaries and that the $n$ pieces are glued again in a different order, specified by a permutation $\pi$ of $\{1,\dots,n\}$: the piece that was in position $j$ is moved to position $\pi(j)$, $j=1,\dots,n$ (see Fig.~\ref{fig:cut}). The new band is described by a new density, which we denote $\rho^\pi$, whose cell masses are a permutation of those of $\rho$:
\begin{equation}\label{eq:rearranged}
m_{\pi(j)}(\rho^\pi)=m_j(\rho),\quad j=1,\dots,n
\end{equation}
and whose profile inside each cell is the profile that $\rho$ had in the piece that has been moved there. We call this operation a \emph{rearrangement} of the band. Note that a rearrangement is a bijection of the band onto itself which preserves lengths (it is a piecewise translation), and that it induces a bijection $\rho\mapsto\rho^\pi$ of the set of densities onto itself, under which genuine densities are mapped to genuine densities, delta components to delta components, and the uniform density to itself.
\begin{figure}[t]
\centering
\begin{tikzpicture}[scale=0.9]
\foreach \j/\c in {1/red!30,2/orange!40,3/yellow!50,4/green!30,5/cyan!30,6/blue!25,7/violet!30,8/gray!40}{
  \fill[\c] ({(\j-1)*1.1},0) rectangle ({\j*1.1},0.5);
  \draw ({(\j-1)*1.1},0) rectangle ({\j*1.1},0.5);
  \node at ({(\j-1)*1.1+0.55},0.25) {\small $\j$};
}
\node[left] at (0,0.25) {$\rho$:};
\foreach \j/\c/\lab in {1/cyan!30/5,2/red!30/1,3/gray!40/8,4/yellow!50/3,5/blue!25/6,6/orange!40/2,7/violet!30/7,8/green!30/4}{
  \fill[\c] ({(\j-1)*1.1},-1.2) rectangle ({\j*1.1},-0.7);
  \draw ({(\j-1)*1.1},-1.2) rectangle ({\j*1.1},-0.7);
  \node at ({(\j-1)*1.1+0.55},-0.95) {\small \lab};
}
\node[left] at (0,-0.95) {$\rho^\pi$:};
\draw[->,>=Latex,thick] (4.4,-0.1) -- (4.4,-0.6);
\end{tikzpicture}
\caption{A rearrangement of the band: the $n=8$ cells of equal length are cut and glued again in a different order. The density $\rho$ of the band becomes the density $\rho^\pi$, whose cell masses are a permutation of those of $\rho$.}
\label{fig:cut}
\end{figure}

A $\U$-measurement selects $\rho$ at random. We say that it is \emph{indifferent} if it does not privilege any portion of the band, i.e., if it ignores which piece of the band is which: selecting $\rho$ or selecting its rearranged version $\rho^\pi$ must be, for an indifferent $\U$-measurement, one and the same thing.

\begin{definition}[universal measurement]\label{def:indifferent}
A $\U$-measurement is \emph{indifferent} if, for every $n$ of the dyadic form $n=2^k$, $k=1,2,3,\dots$, and for every permutation $\pi$ of the $n$ cells \eqref{eq:cells}, the probability law of the random measure $\rho$ is invariant under the rearrangement map $T_\pi:\rho\mapsto\rho^\pi$; equivalently, $\U(T_\pi^{-1}E)=\U(E)$ for every measurable set $E$ of probability measures. In particular, the joint probability law of the $n$ cell masses $(m_1(\rho),\dots,m_n(\rho))$ is symmetric under any permutation of the cells. An indifferent $\U$-measurement will be called a \emph{universal measurement}.\footnote{The restriction to $n=2^k$ is only a matter of convenience: dyadic cells are nested (each cell of level $k$ is the union of two cells of level $k+1$), which simplifies the proofs. Any other sequence of nested partitions into cells of equal length would do, and the condition could equally be imposed for all $n$.}
\end{definition}

\noindent In the language of probability theory, the cell masses of a universal measurement form an \emph{exchangeable} family of random variables, at every level $k$. Physically, Definition~\ref{def:indifferent} is the transcription of the statement ``nothing in the universal measurement allows one to attribute any preference to some regions of the band over others'': a universal measurement is the situation of \emph{maximal lack of knowledge about `the way' the interaction is selected}, in the same way as the uniform $\rho_u$ is the situation of maximal lack of knowledge about `the' interaction itself. This is why we reserve for them the name \emph{universal}.

The $\U_n$-measurement given in \cite{AS2014a} is indifferent at its own level $n$: a rearrangement maps a cellular density with $n^d$ disintegrable cells to another cellular density with the same number of disintegrable cells, so that the uniform law on the $2^n-1$ cellular densities is unchanged. It is therefore not a universal measurement in the strict sense of Definition~\ref{def:indifferent}.\footnote{It is not indifferent at the finer levels $2n,4n,\dots$, since the two halves of a disintegrable cell always carry the same mass; we shall come back to this in Section~\ref{sec:finite}.} Theorem~\ref{thm:finite} will nevertheless show that it produces the universal average exactly, for every $n$. But the definition is much more general: it does not refer to any particular family of densities, finite or infinite, to any discretization, or to any limit. It is a property of the law $\U$ alone.

\begin{theorem}[uniqueness of the universal average]\label{thm:uniqueness} 
Let $\U$ be a universal measurement. Then its average density is the uniform density:
\begin{equation}\label{eq:thm1}
\rhobar_\U(x)=\avgU{\rho(x)}=\rho_u(x)=\tfrac12,
\quad \text{for a.e. }x\in[-1,1]
\end{equation}
Consequently, for every initial state $\rr$,
\begin{equation}\label{eq:thm1-born}
\avgU{\PP(\rr\to\nn_i|\rho)}=\PP(\rr\to\nn_i|\rho_u)=\frac{1+\rr\cdot\nn_i}{2},\quad i=1,2
\end{equation}
In other words, the outcome probabilities of a universal measurement are the quantum probabilities, whatever the universal measurement.
\end{theorem}

\begin{proof}
Let us consider the $n=2^k$ cells of level $k$. The cell masses of the average density are the averages of the cell masses of $\rho$:\begin{equation}\label{eq:avg-masses}
m_j(\rhobar_\U)=\int_{\sigma_j^{(n)}}\avgU{\rho(x)}\,dx=\avgU{m_j(\rho)},\quad j=1,\dots,n
\end{equation}
Since $\U$ is indifferent, the two random variables $m_j(\rho)$ and $m_l(\rho)$, for any two cells $j$ and $l$, take any given range of values with the same probability. Indeed, let $\pi$ be the rearrangement which exchanges the cells $j$ and $l$ (leaving the others in place). By \eqref{eq:rearranged}, the mass of the cell $l$ after the rearrangement is the mass that the cell $j$ had before it, $m_l(\rho^\pi)=m_j(\rho)$; and since $\rho^\pi$ is selected with the same probability as $\rho$, the random variable $m_l(\rho^\pi)$ takes its values with the same probabilities as $m_l(\rho)$. Hence $m_j(\rho)$ and $m_l(\rho)$ take their values with the same probabilities. In particular, this means that they have the same mean value, $\avgU{m_j(\rho)}=\avgU{m_l(\rho)}$. So the $n$ numbers $m_j(\rhobar_\U)$ are all equal, and since they sum to $1$, each of them is equal to $1/n$:
\begin{equation}\label{eq:equal-masses}
m_j(\rhobar_\U)=\frac1n=\frac{\text{length of }\sigma^{(n)}_j}{2},\quad j=1,\dots,n
\end{equation}
In other words, the average density gives to every cell, of every level, a mass equal to its length divided by the total length of the band, which is what the uniform density does. It remains to see that this forces $\rhobar_\U$ to be the uniform density everywhere, not only on the cells. For this, we use the cumulative function \eqref{eq:F} of $\rhobar_\U$, which we denote $\bar F(x)$. By \eqref{eq:equal-masses}, at every dyadic point $x=-1+j\frac{2}{2^k}$ of the band (the right end of the $j$-th cell of level $k$), $\bar F(x)$ is the sum of the masses of the first $j$ cells, so that
\begin{equation}\label{eq:F-dyadic}
\bar F(x)=\frac{j}{2^k}=\frac{1+x}{2}=F_u(x)\quad\text{at every dyadic point }x
\end{equation}
Take now an arbitrary point $x\in[-1,1]$ and, for each level $k$, the two consecutive dyadic points $x'_k\le x\le x''_k$ of that level, which are at distance $2/2^k$ from each other. Monotonicity of the cumulative function gives $F_u(x'_k)=\bar F(x'_k)\le\bar F(x)\le\bar F(x''_k)=F_u(x''_k)$, and letting $k\to\infty$, since $F_u$ is continuous, $\bar F(x)=F_u(x)=\frac{1}{2}(1+x)$ for every $x$. Two probability measures with the same cumulative function are identical. In particular, two ordinary densities with the same cumulative function coincide almost everywhere. Hence $\rhobar_\U=\rho_u$ (as probability measures, or as densities almost everywhere).
\end{proof}

\noindent It is instructive to restate the argument as a proof by contradiction, as this is the form in which it directly expresses the physical intuition. Suppose that the average density of a universal measurement were \emph{not} uniform. Then, since the dyadic cells determine the density, there would be a level $k$ and a cell receiving an average mass different from $1/n$, hence two cells $j$ and $l$ of that level receiving different average masses, $\avgU{m_j(\rho)}\neq\avgU{m_l(\rho)}$. But the universal measurement is by definition indifferent: the rearrangement which exchanges the pieces $j$ and $l$ of the band does not change the law of $\rho$, and therefore does not change the average mass of any cell; applying it, one gets $\avgU{m_j(\rho)}=\avgU{m_l(\rho)}$, a contradiction. In other words, a non-uniform average density would mean that the universal measurement gives more weight to some piece of the band than to another piece of the same length, whereas by hypothesis it does not distinguish between them.

Theorem~\ref{thm:uniqueness} answers the question (Q1). The proof uses two ingredients only: the symmetry of the law $\U$ under rearrangements, and the linearity of the outcome probabilities in $\rho$. It is completely insensitive to the kind of densities that $\U$ selects (two-valued, many-valued, continuous, with or without delta components), to the number of them (finite or infinite), to the weights given to them (which must be equal only for densities that are rearrangements of one another, and are otherwise arbitrary; see Section~\ref{sec:finite}), to the resolution, and to the existence of a limit. Of the ``three constraints'' of the construction in \cite{AS2014a}, only one matters: the cells have the same \emph{length} and are treated symmetrically. This is the constraint that encodes indifference. The other two were conveniences of presentation. As for the question (Q4), the theorem makes explicit where the length enters: in the very definition of indifference, through the cells of equal length. We shall see in Section~\ref{sec:length} that this is not a defect but a necessity.

The theorem also makes precise the sense in which the universal average is \emph{unique}: the average density is the same, $\rho_u$, for \emph{all} universal measurements. There is therefore no need to specify a particular ``way of choosing $\rho$ at random among all the possible $\rho$'': all the indifferent ways give the same answer. This is what dissolves the regress that one might fear when ``averaging over all the ways of selecting an interaction'': the choice of the way of averaging (which could seem to call for a further average over all the ways of averaging, and so on) is immaterial, as long as it is indifferent. We shall write, accordingly,
\begin{equation}\label{eq:univ-new}
\avg{\PP(\rr\to\nn_i|\rho)}^{\rm univ}\equiv\avgU{\PP(\rr\to\nn_i|\rho)},\quad\U\text{ any universal measurement}
\end{equation}
which extends and makes independent of the construction the definition \eqref{eq:univ2014} of \cite{AS2014a}. 

A word on terminology is in order, since Definition~\ref{def:indifferent} sharpens the usage of \cite{Aerts1998,Aerts1999,AS2014a,AS2015a}, where a universal measurement was described as the random choice among \emph{all} the possible measurements, and was implemented through the specific limiting construction recalled in Section~\ref{sec:2014}. We have replaced that construction by the symmetry it was meant to realize. It is indifference that carries with it the label ``universal,'' being the exact transcription of the ``condition of maximal lack of knowledge'' invoked in \cite{Aerts1998}. And it is Theorem~\ref{thm:uniqueness} that justifies the definite article in ``\emph{the} universal average'': universal measurements are infinitely many as laws, but they all assign the same probabilities to the outcomes, so that, as far as the outcome of a single measurement is concerned, they are one and the same measurement. This is, we believe, the precise form of the intuition expressed in \cite{Aerts1999}, that if there is `one' physical reality, then there should also be `one' universal measurement. We shall see in Section~\ref{sec:sequences}, however, that universal measurements sharing the same outcome probabilities can still differ in the correlations they produce between \emph{repeated} measurements, so that the uniqueness established here is the uniqueness of the statistics of a single measurement.

Before proceeding, a word on a possible objection. Theorem~\ref{thm:uniqueness} is a symmetry theorem, and like all symmetry theorems it may look, once stated, almost tautological: an average which treats all pieces of the band alike gives a density which treats all pieces of the band alike, and this is the uniform density. In a sense this is true, and it is also the whole point. What is not tautological is that the result holds for \emph{every} indifferent selection law, without privileging any family of densities, prescribing any particular way of weighing them, or requiring a limiting procedure. Moreover, indifference is a genuine restriction, not a reformulation of the Born rule: as we shall see in Section~\ref{sec:length}, some random selections are not indifferent and do not yield the Born rule, while others may yield it without being indifferent. And Section~\ref{sec:rigidity} will reveal the remarkable rigidity of this condition: no non-trivial indifferent selection of densities without deterministic components is possible. Finally, the theorem also clarifies what indifference does \emph{not} determine: the notion of length on the band (Section~\ref{sec:length}), while the derivation of the Born probabilities additionally requires the orthogonal fall of the particle (Section~\ref{sec:general}). What the universal average achieves is therefore precise: it replaces the direct hypothesis of uniform disintegration by a symmetry-based hypothesis of indifference at the level of the selection of the measurement, and proves that the latter implies the former. That this implication may seem almost obvious once the relevant symmetry has been identified does not make it devoid of content.

\section{Indifference at finite resolution}\label{sec:finite}

Theorem~\ref{thm:uniqueness} tells us what happens when a $\U$-measurement is indifferent at all resolutions, i.e., when it is universal. The $\U_n$-measurements of \cite{AS2014a}, however, live at a single, finite resolution: they select among densities which are constant inside each of the $n$ cells of a given level, and they are indifferent with respect to the rearrangements of these $n$ cells only. It is therefore natural to ask what indifference gives at a single level, and the answer turns out to be even simpler than Theorem~\ref{thm:uniqueness}: the uniform density is obtained \emph{exactly}, at every finite $n$, with no limit to be taken, and this for every indifferent way of selecting cellular densities, not only for the uniform selection among two-valued ones adopted in \cite{AS2014a}. The reason is that, at a given resolution, the densities can be grouped into families of densities that are rearrangements of one another, and that the average of each family is already the uniform density. This is the key to the question (Q2), and it also explains why the proof of \cite{AS2014a} could be carried out by induction.

We fix $n$ and consider the set of all densities which are constant inside each of the $n$ cells \eqref{eq:cells}: we call them \emph{cellular densities of level $n$}, in a sense slightly more general than that of \cite{AS2014a}, since we now allow arbitrary values in the cells. Such a density is completely specified by its $n$ cell masses $p=(p_1,\dots,p_n)$, $p_j\ge0$, $\sum_jp_j=1$:
\begin{equation}\label{eq:cellular-general}
\rho_p(x)=\frac n2\sum_{j=1}^np_j\,\ind_{\sigma_j^{(n)}}(x)
\end{equation}
where $\ind_{\sigma_j^{(n)}}$ is the characteristic function of $\sigma_j^{(n)}$, and $\frac n2\ind_{\sigma_j}$ is the uniform density on the $j$-th cell. The two-valued cellular densities \eqref{eq:cellular} are those for which the $p_j$ take only the values $0$ and $1/n^d$. A rearrangement $\pi$ of the cells maps $\rho_p$ to $\rho_{p^\pi}$, with $p^\pi_{\pi(j)}=p_j$. We can then define the \emph{orbit} of $\rho_p$ as the set of the $n!$ densities $\rho_{p^\pi}$ obtained from $\rho_p$ by all possible rearrangements (some of which may coincide).

\begin{theorem}[finite resolution]\label{thm:finite}
For every cellular density $\rho_p$ of level $n$, the average of its orbit is the uniform density:
\begin{equation}\label{eq:orbit}
\frac{1}{n!}\sum_{\pi}\rho_{p^\pi}=\rho_u
\end{equation}
Consequently, for every $\U$-measurement $\tilde{\U}_n$ which selects cellular densities of level $n$ in a way that is invariant under the rearrangements of the $n$ cells, we have $\rhobar_{\tilde{\U}_n}=\rho_u$. In particular, this holds for the $\U_n$-measurement, which proves \eqref{eq:avg-n-uniform} for all $n$ and for all intervals $A_i$, not only those made of whole cells.
\end{theorem}

\begin{proof}
Consider any rearrangement of the $n$ cells. Applying it to the orbit average \eqref{eq:orbit} merely permutes the $n!$ terms of its sum, and the orbit average remains unchanged, hence it is a cellular density invariant under all rearrangements. Its $n$ cell masses are therefore all equal, and since they sum to $1$, each of them is $1/n$, i.e., the orbit average is $\rho_u$.\footnote{The same conclusion can be reached by direct counting. By \eqref{eq:rearranged}, the mass of the $j$-th cell of $\rho_{p^\pi}$ is $p_{\pi^{-1}(j)}$, and for each index $l$ there are exactly $(n-1)!$ permutations $\pi$ with $\pi^{-1}(j)=l$, namely those with $\pi(l)=j$, the other $n-1$ indices being permuted freely; so $\frac{1}{n!}\sum_\pi p_{\pi^{-1}(j)}=\frac{(n-1)!}{n!}\sum_lp_l=\frac1n$.} The law $\tilde{\U}_n$ being invariant under rearrangements, we have $\avg{\rho}_{\tilde{\U}_n}=\avg{\rho^\pi}_{\tilde{\U}_n}$ for every $\pi$, so that averaging over all $\pi$ and using \eqref{eq:orbit}, gives $\avg{\rho}_{\tilde{\U}_n}=\avg{\frac{1}{n!}\sum_\pi\rho^\pi}_{\tilde{\U}_n}=\avg{\rho_u}_{\tilde{\U}_n}=\rho_u$. The law $\U_n$ is clearly a special case of $\tilde{\U}_n$, and since $\rhobar_{\U_n}=\rho_u$ as densities, the probabilities $\PP(x\in A|\rhobar_{\U_n})$ and $\PP(x\in A|\rho_u)$ coincide for every interval $A$, by \eqref{eq:key}.
\end{proof}

\noindent According to Theorem~\ref{thm:finite}, every orbit has the uniform density as its center of mass, so that the average over all of them is uniform whatever the weights given to the different orbits. The proof by induction of \cite{AS2014a} and the combinatorial identities obtained there are just the shadow of this symmetry in the particular family of two-valued densities (see Appendix~\ref{app:comb}).

The answer to (Q2) is now clear. The sequence of averages of \cite{AS2014a}  is constant because the symmetry already fixes the average at every finite resolution; there is nothing to be gained, as far as the \emph{value} of the average is concerned, from the limit $n\to\infty$. The triviality of the constant sequence is not a defect of the result but its very content: the value $\rho_u$ is exact, independent of the resolution, and independent of the family, which is what one expects from a structural fact, like the fact that the center of mass of a symmetric body lies on its axis of symmetry. What the limit $n\to\infty$ was meant to secure in \cite{AS2014a}  is something else, namely that \emph{every} density is reachable, so that the average can be said to be over all of them. This is the question (Q3), to which we now turn.

Every genuine density $\rho$ is the limit, as $n\to\infty$, of its \emph{histogram} of level $n$, i.e., of the cellular density $\rho_{p}$ whose cell masses are those of $\rho$, $p_j=m_j(\rho)$. The convergence holds in the strongest sense that matters for probabilities: since $\int|\rho_p-\rho|dx\to0$ (convergence in the $L^1$ norm, a standard result of analysis), the difference between $\int_A\rho_p$ and $\int_A\rho$ tends to zero \emph{uniformly} over all regions $A$. So the union of all the finite levels, with the general cellular densities \eqref{eq:cellular-general}, exhausts all genuine densities, and one single sequence (the histograms) works for all regions at once. For the two-valued densities of \cite{AS2014a}, which cannot approximate a density in the $L^1$ sense (they take only two values), one can still construct, for every density $\rho$, a single sequence of two-valued cellular densities $\rho_n$ such that $\PP(x\in A|\rho_n)\to\PP(x\in A|\rho)$ for every interval $A$, with an explicit rate. This is done in Appendix~\ref{app:approx}, and it settles the quantifier issue raised by Shackel about the first step of the proof in \cite{AS2014a}.

But the point we want to stress is that, in the present formulation, the computation of the universal average no longer rests upon this approximation property. The value of the universal average is fixed by Theorems~\ref{thm:uniqueness} and \ref{thm:finite}, at every resolution and for every indifferent law, without any limit to interchange with an average. The approximation property only tells us that the family of finite resolutions leaves no density out. And in the next section we shall see that the words ``all the possible $\rho$'' can be given an even more literal meaning, not through a limit of finite families, but through a single universal measurement which reaches all densities at once.

\section{Mixtures of deterministic measurements}\label{sec:pure}

The original idea of \cite{Aerts1998,Aerts1999} was that a universal measurement is a random choice among \emph{all} imaginable types of measurement: ``deterministic classical measurements, quantum measurements, super-quantum and sub-quantum measurements.'' Deterministic measurements are, in the present language, those of the form $\rho(x)=\delta(x-z)$, in which the band starts disintegrating at a predetermined point $z$. They are the elementary building blocks of all the others, a general $\rho$ being a continuous combination $\rho(x)=\int\delta(x-z)\rho(z)dz$ of deterministic measurements. We now show that, once deterministic measurements are admitted, universal measurements exist in abundance, with a transparent physical meaning, and that among them there are some whose support is the set of \emph{all} densities. We shall also see, in Section~\ref{sec:rigidity}, that deterministic measurements are not only natural but \emph{unavoidable}: there is no non-trivial universal measurement which selects only genuine densities. 

Let us consider three paradigmatic examples, and start with a \emph{random deterministic measurement}. Let $\rho(x)=\delta(x-z)$, $z\in [-1,1]$. A random deterministic measurement is a  $\U$-measurement which, at each run, selects a deterministic measurement with a uniformly distributed disintegrable point $z$. It is obviously indifferent and its average density is $\avgU{\delta(x-z)}=\frac12$, the uniform density, in agreement with Theorem~\ref{thm:uniqueness}. But this universal measurement is nothing but the original hidden-measurement model itself, read at the meta-level: ``first select at random, uniformly, a deterministic measurement, then perform it'' is operationally the same as ``perform the uniform measurement.'' Thus, \emph{the uniform measurement is the simplest universal measurement}, and this is, in a sense, the content of the original conjecture of \cite{Aerts1998} in its most elementary form.

Consider now a finite number $k$ of disintegrable points $z_1,\dots,z_k$, chosen independently and uniformly in $[-1,1]$, and let $\rho(x)=\frac1k\sum_{i=1}^k\delta(x-z_i)$. This corresponds to a band which can only disintegrate at $k$ randomly located points, with equal probabilities. Again, $\U$ is indifferent, and $\rhobar_\U=\rho_u$. As $k\to\infty$, the empirical density of the $k$ points approaches the uniform density (law of large numbers), so that the random $k$-band becomes, in the limit, the uniform band.

Finally, let us take infinitely many disintegrable points, $z_1,z_2,\dots$, independent and uniform in $[-1,1]$, with random weights $w_1,w_2,\dots$, $\sum_iw_i=1$, $w_i\ge0$, chosen independently of the positions by the following ``stick-breaking'' recipe \cite{Sethuraman}: take a stick of length $1$ and break off a random fraction $V_1$ of it: the piece broken off has length $w_1=V_1$. Then break off a random fraction $V_2$ of what remains (which has length $1-V_1$): the piece broken off has length $w_2=(1-V_1)V_2$. Continue in the same way: at the $i$-th step a fraction $V_i$ of the remaining stick, of length $\prod_{l<i}(1-V_l)$, is broken off, so that the $i$-th piece has length $w_i=V_i\prod_{l<i}(1-V_l)$. The lengths $w_i$ of the pieces are the weights: they are positive and, since the whole stick is eventually used up, they sum to $1$. The fractions $V_i$ are independent random numbers in $[0,1]$, all with the same probability density
\begin{equation}\label{eq:beta-law}
f(v)=\alpha\,(1-v)^{\alpha-1},\quad 0\le v\le1
\end{equation}
where $\alpha>0$ is a fixed parameter (in the language of statistics, the $V_i$ follow a Beta$(1,\alpha)$ law, and $\alpha$ is called the \emph{concentration} parameter). For $\alpha=1$ the fractions are uniformly distributed; for small $\alpha$ the density \eqref{eq:beta-law} favors fractions close to $1$, so that a few points receive almost all the weight; for large $\alpha$ it favors small fractions, so that the weight is spread over many points. Then, the random density
\begin{equation}\label{eq:dirichlet}
\rho(x)=\sum_{i=1}^\infty w_i\,\delta(x-z_i)
\end{equation}
describes a band with a countable set of disintegrable points of random strengths (Fig.~\ref{fig:dirichlet}). Its law is called, in probability theory and Bayesian statistics, the \emph{Dirichlet process} with concentration $\alpha$ and ``base'' probability density $\rho_u$ (in the standard notation it is written $\mathrm{DP}(\alpha\rho_u)$, the two parameters being combined in a single measure of total mass $\alpha$) \cite{Ferguson,BlackwellMacQueen,GvdV}. It is indifferent, since the positions are uniform and the weights are chosen independently of the positions. So, Theorem~\ref{thm:uniqueness} applies and its average density is $\rho_u$. This is also a well-known property of the Dirichlet process: its mean is the base probability density, here $\rho_u$, whatever the value of $\alpha$. The parameter $\alpha$ interpolates between the two previous examples: for $\alpha\to0$ the first fraction $V_1$ is close to $1$, almost all the weight goes to the first point and one recovers the random deterministic measurement considered above; for $\alpha\to\infty$ the weights become small and numerous, and the random density becomes indistinguishable, as far as outcome probabilities are concerned, from the uniform density, as in the previous limit $k\to\infty$.
\begin{figure}[t]
\centering
\begin{tikzpicture}[scale=1.0]
\draw[->] (-4.3,0) -- (-4.3,2.2) node[above] {$w_i$};
\draw[dashed,gray] (-4,0) -- (4,0);
\foreach \x/\h in {-3.6/0.12, -3.1/1.6, -2.4/0.5, -1.7/0.25, -1.2/0.1, -0.9/1.1, -0.2/0.15, 0.1/0.3, 0.6/0.7, 1.3/0.4, 1.9/0.1, 2.2/0.2, 2.6/0.9, 3.4/0.3, 3.9/0.08}{
  \draw[thick,MidnightBlue] (\x,0) -- (\x,\h);
  \fill[MidnightBlue] (\x,\h) circle (1.8pt);
  \draw[thin,gray!60] (\x,0) -- (\x,-0.55);
  \fill[MidnightBlue] (\x,-0.7) circle (1.5pt);
}
\node[MidnightBlue,above] at (-3.1,1.65) {$w_1$};
\node[MidnightBlue,above] at (-0.9,1.15) {$w_2$};
\node[MidnightBlue,above] at (2.6,0.95) {$w_3$};
\node[MidnightBlue,above] at (0.6,0.75) {$w_4$};
\draw[line width=3pt,gray] (-4,-0.7) -- (4,-0.7);
\fill (-4,-0.7) circle (2.2pt) node[below=3pt] {$\nn_2$};
\fill (4,-0.7) circle (2.2pt) node[below=3pt] {$\nn_1$};
\node[MidnightBlue,below] at (-3.1,-0.8) {$z_1$};
\node[MidnightBlue,below] at (-0.9,-0.8) {$z_2$};
\node[MidnightBlue,below] at (2.6,-0.8) {$z_3$};
\node[MidnightBlue,below] at (0.6,-0.8) {$z_4$};
\node at (-4.3,-0.7) [left] {$x=-1$};
\node at (4.3,-0.7) [right] {$x=1$};
\end{tikzpicture}
\caption{A realization of the Dirichlet universal measurement \eqref{eq:dirichlet}: the band (bottom) can only break at a countable set of randomly located points $z_i$, each with a random strength $w_i$. These weights sum to $1$ and only the four largest are labeled in the figure, the others being too small to be visible at this scale. Although each realization is a mixture of  deterministic measurements, the average over the realizations is the uniform, quantum, measurement.}
\label{fig:dirichlet}
\end{figure}

What makes the Dirichlet process remarkable for our purposes is that, in an operational sense, it leaves no density out. To explain what we mean, observe that a universal measurement is a random selection, so that the question ``which densities can it select?'' is a question about probabilities. Densities can be compared through the outcome probabilities they produce, so let us say that a density $\rho'$ \emph{imitates} a density $\rho$ to within a tolerance $\varepsilon$, for a given finite family of initial states with $\xr\in(-1,1)$, if $|\PP(x\in A|\rho')-\PP(x\in A|\rho)|<\varepsilon$ for each of the corresponding outcome intervals $A_1$. The Dirichlet universal measurement has the following property \cite[Section~3]{Ferguson}: \emph{whatever} the density $\rho$ on $[-1,1]$, genuine or with delta components, whatever the tolerance $\varepsilon>0$, and whatever the finite family of initial states, the probability that the random density \eqref{eq:dirichlet} imitates $\rho$ to within $\varepsilon$ is strictly positive.\footnote{For any finite partition of the band into cells of positive length, the probabilities assigned to these cells by the Dirichlet process follow a Dirichlet distribution with strictly positive parameters, and can therefore approximate any prescribed probability vector with non-zero probability. More generally, the Dirichlet process has full support in the topology of weak convergence \cite{Ferguson,Majumdar1992}. This does not imply uniform approximation over all intervals when the target density contains delta components.} A prescribed genuine density is never selected exactly (a sum of delta functions never coincides with a genuine density), but every density can be imitated as closely as one wishes, for any prescribed finite family of initial states, and this with a non-zero probability. In the language of probability theory, one says that the Dirichlet process has \emph{full weak support} on the set of all probability distributions on $[-1,1]$. In this precise sense, the Dirichlet universal measurement can approximate the outcome statistics of \emph{all} possible ways of selecting an interaction, none excluded, and the universal average over all of them is the quantum measurement.

This answers (Q3) in a precise operational sense: a single universal measurement can approximate the outcome probabilities of any prescribed density, at any fixed finite resolution, and Theorem~\ref{thm:uniqueness} guarantees that its average is the Born rule, as it is for every other universal measurement.

We conclude this section by noting that universal measurements can be combined in two ways. First, if $\U$ and $\U'$ are indifferent, so is the $\U$-measurement which uses $\U$ with probability $q$ and $\U'$ with probability $1-q$: the universal measurements thus form a \emph{convex family}, which is very large. Second, one can mix the selected densities themselves: if $\rho$, $\rho'$ and a random number $\tau\in[0,1]$ are selected independently, $\rho$ according to $\U$ and $\rho'$ according to $\U'$, then the $\U$-measurement which selects $\rho''=\tau\rho+(1-\tau)\rho'$ is again indifferent (a rearrangement acts on $\rho$ and $\rho'$ separately, and leaves their laws unchanged). Theorem~\ref{thm:uniqueness} says that the average density is \emph{constant} on this whole family, equal to $\rho_u$: the quantum measurement is, so to speak, the center of mass of every member of the family. It is worth emphasizing that indifference is a sufficient, but not a necessary, condition for obtaining the uniform average density. Non-universal measurements can also reproduce the Born rule, as we shall illustrate in Section~\ref{sec:length}. What indifference provides is a structural guarantee: the uniform average follows from symmetry alone, independently of the particular selection law.

\section{Rigidity of the universal measurements}\label{sec:rigidity}

We have seen that a universal measurement reproduces the quantum probabilities \emph{on average}. We now show something stronger: the fluctuations of a universal measurement around its quantum average are entirely due to its deterministic (delta) components. If there are none, there are no fluctuations at all and the universal measurement \emph{is} the quantum measurement, with certainty.

Consider a universal measurement $\U$, and the cell masses $m_j(\rho)$ of the random density $\rho$, at a dyadic level $n=2^k$. By Theorem~\ref{thm:uniqueness}, their mean values are $\avgU{m_j}=1/n$. How much do they fluctuate? The answer is given by the following theorem, which involves the \emph{average coincidence probability}
\begin{equation}\label{eq:beta}
\beta\equiv\Big\langle\sum_i w_i^2\Big\rangle_{\!\U}\in[0,1]
\end{equation}
where the $w_i$ are the random weights of the delta components of the selected density $\rho$ (their number and their positions being random as well). More precisely, if $\rho=\rho_c+\sum_iw_i\delta(x-z_i)$, with $\rho_c$ the non-deterministic (non-atomic) component of $\rho$ (possibly including a singular-continuous part), $\sum_iw_i^2$ is the probability that two disintegrable points selected independently according to $\rho$ coincide exactly, since the non-deterministic component $\rho_c$ never produces coincidences. It vanishes if and only if $\rho$ has no deterministic components. The number $\beta$ is the average of this quantity over the random selection of $\rho$, i.e., the probability that two interactions coincide when $\rho$ is first selected according to $\U$ and then kept fixed for both independent selections of the interaction. Hence, it is the \emph{coincidence probability} of the universal measurement $\U$.

Note that $\beta$ is not the total weight $\langle\sum_iw_i\rangle_\U$ of the deterministic components, i.e., the probability that the selected interaction comes from a delta component. Since $\sum_iw_i^2\le\sum_iw_i$, with equality only when all the weight is carried by a single point, $\beta$ weighs the deterministic components quadratically, and is therefore sensitive to their \emph{concentration} rather than to their amount. What matters for us is that $\beta$ vanishes if and only if the universal measurement has no deterministic (delta) components at all, and that, as we now show, it is precisely $\beta$ which governs the fluctuations of the outcome probabilities around their quantum values.

\begin{theorem}[variance of a universal measurement]\label{thm:variance}
Let $\U$ be a universal measurement. Then, for every dyadic level $n$ and every cell $j$,
\begin{equation}\label{eq:var}
\Var_\U\big[m_j(\rho)\big]=\beta \Big(\frac{n-1}{n^2}\Big),\quad \Cov_\U\big[m_j(\rho),m_l(\rho)\big]=-\frac{\beta}{n^2}\quad(j\neq l)
\end{equation}
and more generally, for every $x\in[-1,1]$, writing $p_x\equiv F_u(x)=\frac{1+x}{2}$
\begin{equation}\label{eq:varF}
\Var_\U\big[F_\rho(x)\big]=\beta\,p_x\,(1-p_x)
\end{equation}
\end{theorem}

\begin{proof}
Let $v_n$ be the variance of $m_j(\rho)$ and $c_n$ the covariance of $m_j(\rho)$ and $m_l(\rho)$, $j\neq l$. By indifference, these numbers do not depend on $j$ and $l$ (any two cells, or pairs of cells, can be exchanged by a rearrangement, and variances and covariances only depend on the probabilities with which the masses take their values, as in the proof of Theorem~\ref{thm:uniqueness}). All mean values, variances and covariances in this proof are taken with respect to $\U$, the selected density $\rho$ being the only random quantity. In addition, the variance of a sum of random variables is the sum of their variances plus the sum of the covariances of all the pairs of distinct variables:
\begin{equation}
\Var_\U\Big(\sum_{j=1}^nm_j\Big)=\sum_{j=1}^n\Var_\U(m_j)+\sum_{j\neq l}\Cov_\U(m_j,m_l)=nv_n+n(n-1)c_n
\label{varianceformula}
\end{equation}
where for the second equality we have used that there are $n$ variances, all equal to $v_n$, and $n(n-1)$ ordered pairs $(j,l)$ with $j\neq l$, all with covariance $c_n$. But $\sum_jm_j(\rho)=1$, for every $\rho$, so that this sum has zero variance, and therefore
\begin{equation}\label{eq:c}
c_n=-\frac{v_n}{n-1}
\end{equation}
In other words, since the cell masses always sum to $1$, a cell which is heavier than average forces the others to be lighter, which is why the covariance between two cells is negative, and exactly compensates the variances. 

At this point, we can use the fact that each cell of level $k$ is the union of two cells of level $k+1$, so $m_j$ (level $n$) is the sum of two cell masses of level $2n$. So, using again (\ref{varianceformula}), together with (\ref{eq:c}) applied at level $2n$, i.e., $c_{2n}=-v_{2n}/(2n-1)$, we have 
\begin{equation}
v_n=2\, v_{2n}+2\, c_{2n}=2\, v_{2n}\Big(1-\frac{1}{2n-1}\Big)=2\, v_{2n}\,\frac{2n-2}{2n-1}
\label{relation-v}
\end{equation}
Defining $\beta_n\equiv n^2v_n/(n-1)$, we can write $v_n=\beta_n(\frac{n-1}{n^2})$ and (\ref{relation-v}) gives 
\begin{equation}
\beta_n=\frac{n^2}{n-1}2v_{2n}\frac{2(n-1)}{2n-1}=\frac{(2n)^2v_{2n}}{2n-1}=\beta_{2n}
\end{equation}
So, we have shown that the number $\beta_n$ is the same at all dyadic levels. We can call it $\beta$ and observe that \eqref{eq:var} holds, with this $\beta$, by \eqref{eq:c}.

As a final step, consider the mean value, with respect to $\U$, of $\sum_jm_j(\rho)^2$. By definition of the variance, $\avgU{m_j(\rho)^2}=\Var_\U(m_j)+\avgU{m_j(\rho)}^2=v_n+\frac{1}{n^2}$, so that
\begin{equation}\label{eq:sumsq}
\Big\langle\sum_{j=1}^nm_j(\rho)^2\Big\rangle_{\!\U}=n\Big(v_n+\frac{1}{n^2}\Big)=\beta\Big(\frac{n-1}{n}\Big)+\frac1n=\beta+\frac{1-\beta}{n}
\end{equation}
On the other hand, for a given $\rho$, $\sum_jm_j(\rho)^2$ is the probability that two disintegrable points selected independently according to $\rho$ fall in the same cell of level $k$. As $k$ increases, the cells shrink and this probability decreases to the probability that the two points \emph{coincide}, which is $\sum_iw_i^2$. Taking the mean value over $\U$ (the quantities are bounded by $1$, so mean and limit can be exchanged) and letting $n\to\infty$ in \eqref{eq:sumsq}, we get $\beta=\avgU{\sum_iw_i^2}$, which is \eqref{eq:beta}. Finally, \eqref{eq:varF} follows from \eqref{eq:var} at dyadic points $x$, since $F_\rho(x)$ is then a sum of $q_x=p_xn$ cell masses, whose variance is $q_x v_n+q_x(q_x-1)c_n=\beta\frac{q_x}{n}(1-\frac{q_x}{n})$, and by continuity we can extend to arbitrary $x$.
\end{proof}

\noindent To see what \eqref{eq:varF} means, consider the following simple universal measurement $\U_\beta$: with probability $\beta$ it selects the random deterministic measurement $\rho(x)=\delta(x-z)$, $z$ uniform in $[-1,1]$, and with probability $1-\beta$ it selects the uniform density $\rho_u$. It is indifferent, being a mixture of two universal measurements.  For an initial state such that the particle sits on the band at the coordinate $x$, the outcome probability $F_\rho(x)$ is: in the first case, $1$ if $z\le x$ and $0$ otherwise, i.e., a yes--no variable with mean $p_x$ and variance $p_x(1-p_x)$; in the second case, $p_x$, with variance $0$. In both cases the mean is $p_x$, so the mean of $F_\rho(x)$ is $p_x$, as it must be by Theorem~\ref{thm:uniqueness}, and the variance is the weighted sum of the two variances, $\beta\,p_x(1-p_x)+(1-\beta)\cdot0=\beta\,p_x(1-p_x)$. Theorem~\ref{thm:variance} says that \emph{every} universal measurement with coincidence probability $\beta$ has exactly this variance, whatever its actual structure: as far as the fluctuations of the outcome probabilities around their quantum values are concerned, a universal measurement cannot be distinguished from a random measurement $\U_\beta$ which is deterministic with probability $\beta$ and quantum with probability $1-\beta$.\footnote{More generally, the same parameter $\beta$ determines all second-order correlations of the outcome probabilities: $\Cov_\U[F_\rho(x),F_\rho(y)]= \beta(p_{\min\{x,y\}}-p_xp_y)$. This follows from the covariances of the dyadic cell masses in \eqref{eq:var}, by approximation. Thus, all universal measurements with the same $\beta$ have identical second-order statistics, although their higher-order statistics may differ.}

For the uniform measurement, $\beta=0$. For the random deterministic measurement of Section~\ref{sec:pure}, the selected density has a single delta component of weight $1$, so $\beta=1$: the selected measurement is always deterministic, and $F_\rho(x)$ takes only the values $0$ and $1$, with mean $p_x$ and variance $p_x(1-p_x)$. For $k$ equally likely disintegrable points, each has weight $1/k$, so $\sum_iw_i^2=k /k^2=1/k$, and $\beta=1/k$. For the Dirichlet process with concentration parameter $\alpha$, the stick-breaking construction gives $\beta$ by a self-similarity argument: writing $S=\sum_{i\ge1}w_i^2$, the first piece contributes $w_1^2=V_1^2$, and the remaining pieces are obtained by breaking the remaining stick, of length $1-V_1$, in exactly the same way as the whole stick, so that $S=V_1^2+(1-V_1)^2S'$, where $S'$ has the same law as $S$ and is independent of $V_1$. Taking mean values, $\avg{S}=\avg{V_1^2}+\avg{(1-V_1)^2}\avg{S}$. For the density \eqref{eq:beta-law} one computes $\avg{V_1^2}=\frac{2}{(1+\alpha)(2+\alpha)}$ and $\avg{(1-V_1)^2}=\frac{\alpha}{2+\alpha}$, hence $\avg{S}\big(1-\frac{\alpha}{2+\alpha}\big)=\frac{2}{(1+\alpha)(2+\alpha)}$, i.e.,
\begin{equation}
\beta=\avg{S}=\frac{1}{1+\alpha}
\end{equation}
As expected, $\beta\to1$ for $\alpha\to0$ (random deterministic measurement) and $\beta\to0$ for $\alpha\to\infty$ (uniform measurement).

\begin{theorem}[uniform or deterministic]\label{thm:dichotomy}
Let $\U$ be a universal measurement. Then $\U$ is a mixture of the uniform measurement and of a universal measurement which only selects densities with deterministic components: with probability $q$ the selected density is $\rho_u$, and with probability $1-q$ it has at least one delta component.
\end{theorem}

\begin{corollary}\label{cor:genuine}
In particular ($q=1$): if $\U$ never selects a density with delta components, then it selects, with probability one, the uniform density $\rho_u$.
\end{corollary}

\begin{proof}
We first prove the corollary, then the theorem. If $\rho$ almost surely has no delta components, then $\sum_iw_i^2=0$ almost surely, so $\beta=0$ by \eqref{eq:beta}, so all the variances \eqref{eq:var} vanish: every cell mass $m_j(\rho)$, of every dyadic level, is almost surely equal to its mean value $1/n$. There are countably many dyadic cells, so with probability one \emph{all} of them have mass $1/n$ simultaneously, and by the monotonicity argument of the proof of Theorem~\ref{thm:uniqueness}, $F_\rho=F_u$, i.e., $\rho=\rho_u$. Now, the property ``$\rho$ has no delta components'' is unaffected by rearrangements. Hence, if we condition the law $\U$ on this property, we obtain again a universal measurement, which now selects no delta components, and which by the corollary is the uniform measurement. Conditioning on the complementary property, we obtain a universal measurement selecting densities with delta components only.
\end{proof}

\noindent Corollary~\ref{cor:genuine} is a \emph{rigidity} statement: indifference, together with the absence of deterministic components, leaves no freedom at all. It is the strongest possible form of the intuition that guided the analysis in \cite{AS2014a}, that ``averaging over all densities cannot give anything but the uniform density'': if the random selection does not privilege any portion of the band and selects genuine densities, then it does not merely give the uniform density \emph{on average}, it gives the uniform density \emph{every time}. In the hidden-measurement language: a universal measurement without deterministic components is, with certainty, the quantum measurement. The Born rule then governs the outcome probabilities of each individual run, and not merely the statistics obtained by averaging over many runs.

The theorem also explains, retrospectively, a historical difficulty. The attempts made in the nineties to define a universal average by means of a probability law on the set of all densities (see the Preface of \cite{ASbook}) could not succeed, not because of technical subtleties, but because \emph{no such law exists}, except the trivial one concentrated on $\rho_u$. Non-trivial universal measurements exist only if one admits deterministic measurements, as in Section~\ref{sec:pure}. In probability theory this is a known structural fact: a random probability measure on $[-1,1]$ whose law is invariant under the length-preserving transformations of the interval is almost surely of the form $\tau\rho_u+\sum_iw_i\delta(x-z_i)$, with the $z_i$ independent and uniform \cite[Theorem~1.25, p.~50]{KallenbergSym}; see also \cite{Kallenberg1973}. 
Theorem~\ref{thm:dichotomy} is the elementary part of this representation, which also shows that the continuous part of the density selected by a universal measurement is always of the form $\tau\rho_u$, i.e., the uniform density multiplied by its total weight $\tau\in[0,1]$, the remaining weight $1-\tau$ being carried by the delta components: no non-uniform profile can appear, not even alongside deterministic components.

What happens to the cellular measurements $\U_n$ introduced in \cite{AS2014a}? They are indifferent only at their own level, and they select two-valued densities, which are genuine densities. So, Theorem~\ref{thm:dichotomy} does not apply to them directly, and indeed a two-valued density is very different from the uniform one (half of the band, on average, is non-disintegrable). Yet, as $n\to\infty$, the \emph{outcome probabilities} of the randomly selected two-valued density become indistinguishable from the quantum ones, with probability tending to one. This is a law of large numbers: the fraction of disintegrable cells in any region $A$ of the band is close to the fraction of disintegrable cells in the whole band, with fluctuations of order $1/\sqrt n$; more precisely (Appendix~\ref{app:concentration})
\begin{equation}\label{eq:concentration}
\Var_{\U_n}\big[\PP(x\in A|\rho_n)\big]\le\frac{2}{n+1}
\end{equation}
for every region $A$ made of whole cells. So the $\U_n$-measurements also \emph{concentrate} on the quantum measurement, not in the sense that the selected density becomes uniform (it does not: in each cell it is either zero or twice the uniform value), but in the sense that its transition probabilities become those of the uniform density. This is the way in which Theorem~\ref{thm:dichotomy} manifests itself at finite resolution, and it clarifies in what sense the two-valued cellular densities ``approach'' the uniform density as $n\to\infty$. They do not approach it as functions,\footnote{A two-valued density of level $n$ takes the value $0$ in about half of the cells and about twice the uniform value in the other half, so that, as $n$ grows, its profile becomes an ever finer zigzag between $0$ and $1$, which converges to no function at all and remains at a distance $\int_{-1}^1|\rho_n-\rho_u|dx\approx1$ from the uniform density.} but their integrals over any fixed interval converge to that of $\rho_u$, because the oscillations compensate each other on scales smaller than the interval. In other words, it is the \emph{probabilities}, not the densities, that become quantum. This is also, in our view, the correct way to read the remark, made in \cite{Shackel2024}, that the limit of the cellular densities would be a set of Dirac delta functions: they do not, they only become infinitely fast oscillations whose local average is uniform.

\section{Sequences of measurements}\label{sec:sequences}

The average coincidence probability $\beta$ in (\ref{eq:beta}) has an additional physical meaning when one considers \emph{sequences} of measurements. Suppose that the same measurement is performed on a collection of entities, all prepared in the same state $\rr$, and let us distinguish two scenarios.

In the first scenario, the random density $\rho$ is selected afresh at each run. Then the outcomes are independent, each with the quantum probabilities \eqref{eq:thm1-born}, whatever the universal measurement $\U$ and whatever the value of its coincidence probability $\beta$: the sequence is exactly the one predicted by quantum mechanics for independent repetitions.

In the second scenario, the density $\rho$ is selected once and then \emph{persists} for the whole sequence (the measuring apparatus is ``the same'' at each run, in the sense that its `way' of selecting interactions has been fixed once and for all, although it remains unknown to us). Then the outcomes are no longer independent. The probability of obtaining, say, the outcome $a_1$ in two consecutive runs is the mean value (we write for simplicity $x\equiv \xr$) of $F_\rho(x)^2$, and by Theorem~\ref{thm:variance}
\begin{equation}\label{eq:sequence}
\PP(a_1\text{ and then }a_1)=\avgU{F_\rho(x)^2}=p_x^2+\beta\,p_x(1-p_x),\quad p_x=\frac{1+x}{2},
\end{equation}
to be compared with the value $p_x^2$ predicted by quantum mechanics for independent outcomes. The outcomes are positively correlated: having obtained $a_1$ makes $a_1$ more likely at the next run, by an amount proportional to $\beta$. The sequence is \emph{exchangeable} (its law is invariant under permutations of the runs) but, for $\beta >0$ and $0<p_x<1$, not independent. For the Dirichlet universal measurement this is the well-known P\'olya-urn structure of \cite{BlackwellMacQueen}, in which each new breaking point is either a previously used one, with a probability proportional to the number of times it has been used, or a fresh uniformly distributed one. At the level of the binary outcomes, this induces a corresponding two-color P\'olya-urn process. 

So, the hypothesis that quantum measurements are universal measurements is compatible with the observed independence of repeated quantum measurements in two ways: either the way of selecting the interaction is itself re-selected at each run, or it is persistent but has no deterministic components, in which case, by Theorem~\ref{thm:dichotomy}, it is the uniform measurement. The two cases are in principle distinguishable through the statistics of sequences, with $\beta$ as the parameter to be determined experimentally -- or, more precisely, to be bounded from above, since  the independence of repeated measurements on identically prepared systems, which is routinely assumed and, as far as we know, never contradicted in quantum laboratories, indicates that, if the way of selecting the interaction persists from one run to the next, $\beta$ must be very small. A dedicated analysis of existing sequential data would be needed to turn this into a quantitative bound.

These two scenarios are of course idealizations, and the physically relevant situation is likely to be intermediate: the way of selecting the interaction, i.e., the density $\rho$, could persist for a certain characteristic time $\tau$, and then be ``reset,'' i.e., selected anew. In that case, whether correlations are observed depends on the latency $\Delta t$ between successive runs. To make this quantitative in the simplest way, suppose that the resets occur at random times, with a mean interval $\textsc{t}$ (a Poisson process). Two runs separated by a time $\Delta t$ then share the same $\rho$ with probability $e^{-\Delta t/\textsc{t}}$, in which case their outcomes are correlated as in \eqref{eq:sequence}, and are otherwise independent, so that, for an initial state with Born probability $p$,
\begin{equation}\label{eq:sequence-tau}
\PP(a_1\text{ and then }a_1)=p^2+\beta\,p(1-p)\,e^{-\Delta t/\textsc{t}}
\end{equation}
The correlations are visible only for latencies shorter than the persistence time $\textsc{t}$, and the independence of repeated quantum measurements, as observed in the laboratory, constrains the product $\beta\, e^{-\Delta t/\textsc{t}}$, not $\beta$ alone. This situation is entirely analogous to the one we discussed in Section~10 of \cite{AS2014a}, one level below, for the measurement interactions themselves: if the interaction $\boldsymbol\lambda$, rather than the way of selecting it, persisted in the measuring device for a characteristic time, successive measurements performed within that time on identically prepared entities would tend to produce the same outcome. This is precisely the hypothesis that was tested by Durt et al.\ \cite{Durt2002}, with an atomic interferometer, by analyzing the temporal statistics of successive measurements, with a negative result: no temporal correlations were observed. 

As we observed in \cite{AS2014a}, this negative result admits different readings: the fluctuations could be too rapid to be detected, or, more fundamentally, the hidden interaction could be created anew at each measurement, in the interplay between the measuring system and the measured entity, in the manner of a symmetry breaking, rather than being a pre-existing fluctuating state of the apparatus. In the extended Bloch representation this is suggested by the fact that the membrane is destroyed by each measurement and has to be re-created for the next one. The same two readings are available one level up, for the density $\rho$: either it is a persistent (if unknown) characteristic of the measuring system, in which case the correlations \eqref{eq:sequence-tau} should appear for sufficiently short latencies -- and the negative result of Durt et al., or a refined version of their experiment, also constrains this scenario, although translating it into a quantitative bound on $\beta e^{-\Delta t/\textsc{t}}$ would require a dedicated analysis -- or the way of selecting the interaction is itself created at each contact between the entity and the measuring system, in which case no correlation is to be expected.\footnote{Interestingly, temporal correlations do not require the same hidden-measurement interaction to persist across successive runs. Even when a new interaction is independently selected at each run, correlations can arise if the underlying density $\rho$ remains unchanged: the outcomes are independent conditional on $\rho$, but generally correlated when averaged over its random selection according to $\U$. Thus, memory effects can emerge at the level of the selection law itself, rather than at the level of the individual interactions.} In the latter reading, a universal measurement is, at each run, a genuinely new ``measurement of measurements,'' and the Born rule for sequences holds exactly, with no constraint on the deterministic content of the individual runs. 

We leave it to future work, and to the experimentalists, to assess which of these readings is the correct one. The point we want to make here is that the notion of universal measurement, far from being a purely formal device, articulates a precise and in principle testable hypothesis about the temporal structure of the measurement process.

\section{The role of the length}\label{sec:length}

We now come to the question (Q4): the role of the length. Indifference was defined through cells of equal \emph{length}, and the uniform density is uniform with respect to the length along the band. Could one dispense with this notion, and define a universal average which would depend only on the set of all densities, with no reference to how the band is measured? The answer is no, and the reason is instructive.

Consider rearrangements of a more general kind: cut the band into pieces of \emph{unequal} lengths, and glue them again in a different order, after stretching or compressing each piece so that it fits into its new place. Such an operation maps every density to a density (a stretched piece carries the same mass over a longer interval), so that the set of all densities is invariant under it, exactly as it is under the length-preserving rearrangements. If the universal average could be defined from the set of all densities alone, with no reference to the length, it would have to be invariant under these generalized rearrangements too.

\begin{theorem}[no universal average without a notion of size]\label{thm:length}
There is no probability density on $[-1,1]$ which is invariant under all the generalized rearrangements (cut into pieces of unequal lengths, stretch or compress, and glue again). Consequently, no $\U$-measurement can be indifferent with respect to all of them.
\end{theorem}

\begin{proof}
Let $\rhobar$ be a density invariant under all generalized rearrangements, and let $I$ and $J$ be any two intervals of positive length inside $[-1,1]$. There is a generalized rearrangement mapping $I$ onto $J$ (cut out $I$, stretch or compress it to the length of $J$, and fit the rest of the band around it); by invariance, $\rhobar$ gives the same mass to $I$ and to $J$. Call this common mass $c$. Take now $J$ strictly inside $I$, with $I\setminus J$ containing an interval of positive length: then $c=\int_I\rhobar=\int_J\rhobar+\int_{I\setminus J}\rhobar\ge2c$, so $c=0$. Every interval of positive length would then carry zero mass, which is incompatible with $\int_{-1}^1\rhobar=1$. If a $\U$-measurement were indifferent with respect to all generalized rearrangements, its average density would be invariant under all of them, by the same computation as in the proof of Theorem~\ref{thm:uniqueness}. So, no such $\U$ exists.
\end{proof}

\noindent Theorem~\ref{thm:length} tells us that \emph{indifference is always indifference with respect to a notion of size}. The set of all possible ways of selecting an interaction, taken as a bare set, has no symmetry that could single out a density: one has to say with respect to what the selection is indifferent. In our model, the answer is given by the physics of the model: the band is a segment of the Bloch ball, which is a Euclidean ball, its length is the Euclidean length, and the ``cells of equal length'' are cells of equal Euclidean length. This is the same input that gives the Born rule for the \emph{uniform} measurement in \cite{AS2014a}: the probability of an outcome is the ratio of two lengths, as per \eqref{eq:born}. The universal average does not create this input, it \emph{inherits} it and shows that nothing else is needed. The Euclidean structure of the extended Bloch sphere is in turn not an arbitrary choice: it is induced by the Hilbert--Schmidt scalar product of operators, which is, up to a factor, the only scalar product invariant under unitary transformations.

This is also the sense in which the dependence of the universal average on the ``parameterization'' of the band, which in the context of Bertrand's paradox is the core of Shackel's objection, is to be understood: a passive change of coordinates $y=h(x)$ does not by itself change the physical length measure (which must be pushed forward with the appropriate Jacobian). What changes the notion of size, and potentially the universal average, is instead imposing a \emph{new} uniform measure $dy$ in place of the transformed Euclidean measure $|dx/dy|\,dy$. In the extended Bloch representation there is one natural parameterization, the Euclidean one, and the universal average is uniquely defined with respect to it. The theorem simply says that this input cannot be derived from the set of densities, but has to come from the physical model, and in our case it does.

The distinction between an indifferent and a non-indifferent random selection of measurements is well illustrated by a calculation made by Aerts in \cite{Aerts1998}, before the notion of universal average was given its present form. There, the intermediate measurements were described by elastic bands uniformly disintegrable in a sub-interval $[d-\epsilon,d+\epsilon]$ and non-disintegrable elsewhere (a two-parameter generalization of the $\epsilon$-model), i.e., by the densities
\begin{equation}
\rho_{\epsilon,d}(x)=\frac{1}{2\epsilon}\ind_{[d-\epsilon,d+\epsilon]}(x),\quad 0\le\epsilon\le1,\quad|d|\le1-\epsilon
\end{equation}
and one asked what the transition probabilities would be if $(\epsilon,d)$ were chosen ``at random,'' uniformly in the triangle $T=\{0\le\epsilon\le1,|d|\le1-\epsilon\}$ of the allowed values. The result (Eq.~(7) of \cite{Aerts1998}) was found to be \emph{different} from the Born rule, by logarithmic correction terms. In the present language, this random choice is a universal measurement $\U_{\epsilon,d}$, and its average density can be computed in closed form. The triangle $T$ has unit area, so that choosing $(\epsilon,d)$ uniformly in $T$ means that the average density is obtained by integrating $\rho_{\epsilon,d}(x)$ over $T$ with the measure $d\epsilon\,dd$:
\begin{equation}\label{eq:f-int}
\rhobar_{\U_{\epsilon,d}}(x)=\iint_T\frac{1}{2\epsilon}\,\ind_{[d-\epsilon,d+\epsilon]}(x)\,d\epsilon\,d d
=\int_0^1\frac{d\epsilon}{2\epsilon}\,\Big|[x-\epsilon,x+\epsilon]\cap[-(1-\epsilon),1-\epsilon]\Big|
\end{equation}
where in the second expression we have first integrated over $d$: for a given $\epsilon$, the condition $x\in[d-\epsilon,d+\epsilon]$ means $d\in[x-\epsilon,x+\epsilon]$, and $d$ is also constrained to $[-(1-\epsilon),1-\epsilon]$, so that the integral over $d$ is the length (denoted $|\cdot|$) of the intersection of these two intervals. Consider $x\ge0$ (the case $x<0$ follows by the symmetry $x\to-x$). For small $\epsilon$, namely $\epsilon\le\frac{1-x}{2}$, the interval $[x-\epsilon,x+\epsilon]$ lies entirely inside $[-(1-\epsilon),1-\epsilon]$, and the length is $2\epsilon$. For $\frac{1-x}{2}<\epsilon\le\frac{1+x}{2}$, the right end is cut, and the intersection is $[x-\epsilon,1-\epsilon]$, of length $1-x$. For $\epsilon>\frac{1+x}{2}$, both ends are cut, and the intersection is the whole $[-(1-\epsilon),1-\epsilon]$, of length $2(1-\epsilon)$. The three contributions to \eqref{eq:f-int} are then
\begin{equation}
\int_0^{\frac{1-x}{2}}d\epsilon=\frac{1-x}{2},\quad
\frac{1-x}{2}\int_{\frac{1-x}{2}}^{\frac{1+x}{2}}\frac{d\epsilon}{\epsilon}=\frac{1-x}{2}\ln\frac{1+x}{1-x},\quad
\int_{\frac{1+x}{2}}^{1}\frac{1-\epsilon}{\epsilon}\,d\epsilon=\frac{x-1}{2}-\ln\frac{1+x}{2}
\end{equation}
and their sum, after a short rearrangement of the logarithms (and extending the result to $x<0$ by symmetry), takes the remarkably simple form
\begin{equation}\label{eq:f-entropy}
\rhobar_{\U_{\epsilon,d}}(x)=-\frac{1+x}{2}\ln\frac{1+x}{2}-\frac{1-x}{2}\ln\frac{1-x}{2}=H\Big(\frac{1+x}{2}\Big)
\end{equation}
where $H(p)=-p\ln p-(1-p)\ln(1-p)$ is the Shannon entropy (in natural units) of a yes-no experiment with probability $p$. So, the average density of Aerts' random choice is the entropy of the Born probability $p=(1+x)/2$ of the state whose particle lands at $x$. It is maximal at the center of the band, $\rhobar(0)=\ln2\approx0.69$, and vanishes at its ends, $\rhobar(\pm1)=0$. It is correctly normalized, since $\int_{-1}^1H(\frac{1+x}{2})dx=2\int_0^1H(p)dp=1$. Integrating \eqref{eq:f-entropy} one obtains the cumulative function
\begin{equation}\label{eq:F-epsd}
\bar F(x)=\int_{-1}^x\rhobar_{\U_{\epsilon,d}}(y)\,dy=\frac{1+x}{2}-\frac{(1+x)^2}{4}\ln\frac{1+x}{2}+\frac{(1-x)^2}{4}\ln\frac{1-x}{2}
\end{equation}
as one checks by differentiation, and for a vector-state corresponding to $x=\cos\alpha$, this is exactly Eq.~(7) of \cite{Aerts1998}.
So the transition probability found in 1998 is, as it must be by \eqref{eq:key}, the transition probability of the $\rho$-measurement with the average density \eqref{eq:f-entropy}, and its deviation from the Born rule is the deviation of this average density from the uniform one.

By Theorem~\ref{thm:uniqueness}, the $\U_{\epsilon,d}$-measurement is therefore \emph{not} indifferent and it is easy to see why. A sub-interval $[d-\epsilon,d+\epsilon]$ chosen uniformly in the triangle $T$ is much more likely to cover the center of the band than its ends, because a point near an end can only be covered by the few, small, intervals that fit there (this is what the three regimes of the computation above express). Hence, the uniform choice of a sub-interval privileges the center of the band: it is a random choice of a measurement which is \emph{not} indifferent with respect to the pieces of the band. The lesson is that ``choosing a measurement at random'' is not, by itself, enough to produce the quantum probabilities: the result depends on how the measurements are selected. Indifference provides a sufficient condition, but not a necessary one. Indeed, consider a $\U$-measurement selecting, with equal probability, either the density $\rho_L$, uniform on $[-1,0]$, or the density $\rho_R$, uniform on $[0,1]$. Clearly, $\rhobar_\U=\frac12(\rho_L+\rho_R)=\rho_u$, so that the Born rule is recovered. Yet this $\U$-measurement is not indifferent: although its law is invariant under the exchange of the two halves of the band, it is not invariant under arbitrary rearrangements of four dyadic cells, which can transform $\rho_L$ into a density different from both $\rho_L$ and $\rho_R$. Thus, reproducing the Born rule on average is a weaker requirement than indifference. What Theorem~\ref{thm:uniqueness} establishes is that indifference guarantees the Born rule without any further specification of the selection law. Furthermore, Corollary~\ref{cor:genuine} shows that the only fully indifferent law supported on absolutely continuous probability densities is the one concentrated on $\rho_u$. Thus, while suitably chosen weights can reproduce the Born rule without indifference, a non-trivial universal measurement necessarily requires a broader family of admissible densities.

\section{Generalizations}\label{sec:general}

In the extended Bloch representation, a non-degenerate measurement with $N$ outcomes is represented by an $(N-1)$-dimensional simplex $\triangle_{N-1}$, with vertices $\nn_1,\dots,\nn_N$, filled with an abstract elastic and disintegrable membrane. The point-particle representing the state falls orthogonally onto the membrane, reaching the point $\rpar=\sum_ir_i^\parallel\nn_i$, which divides the simplex into $N$ convex regions $A_1,\dots,A_N$. If the membrane begins to disintegrate in $A_i$, the particle is drawn to $\nn_i$. If the membrane is uniform, the probability of the outcome $a_i$ is the ratio of the Euclidean, $(N-1)$-dimensional volume of $A_i$ to the volume of the simplex, and this ratio is exactly the Born probability $\operatorname{Tr}D(\rr)P_i$ (Eqs.~(59)--(66) of \cite{AS2014a}). A general $\rho$-measurement is described by a density $\rho(\boldsymbol\lambda)$ on the simplex, and $\PP(\rr\to\nn_i|\rho)=\int_{A_i}\rho(\boldsymbol\lambda)d\boldsymbol\lambda$.

The same symmetry-based construction extends naturally to higher dimensions, once the cells of equal length are replaced by cells of equal volume. A convenient nested sequence of such cells is provided by the iterated \emph{barycentric subdivision} of the simplex: the first subdivision cuts $\triangle_{N-1}$ into $N!$ small simplices of equal volume, corresponding to the orderings of the barycentric coordinates ($\lambda_{\pi(1)}\ge\lambda_{\pi(2)}\ge\dots\ge\lambda_{\pi(N)}$, which are permuted among themselves by the volume-preserving permutations of the coordinates). Each small simplex is then subdivided in the same way, and so on, the cells becoming arbitrarily small.\footnote{At every level, cells of equal volume can be exchanged by piecewise affine, volume-preserving transformations, with a fixed convention for their common boundaries. Since the nested partitions have cells of vanishing diameter, they generate the Borel sets of the simplex.}
A $\U$-measurement is universal if its law is invariant under the rearrangements of the cells of every level, and the proof of Theorem~\ref{thm:uniqueness} goes through unchanged: the average density gives to every cell of every level a mass equal to its relative volume, and since the cells of all levels determine the density (they become arbitrarily small and any region is approximated by unions of them), the average density is the uniform one, $\rhobar_\U=\rho_u$, and
\begin{equation}
\avg{\PP(\rr\to\nn_i|\rho)}^{\rm univ}=\operatorname{Tr}D(\rr)P_i,\quad i=1,\dots,N
\end{equation}
for every universal measurement and every $N\ge2$. Theorems~\ref{thm:finite}, \ref{thm:variance} (with the relative length $1/n$ replaced by the relative volume of the cells), \ref{thm:dichotomy} and \ref{thm:length} also carry over, with the same proofs. The $(N-1)$-dimensional cellular construction of Section~13.3 of \cite{AS2014a} is, again, one particular universal measurement, and the Dirichlet process with base probability density the uniform density on the simplex is another, whose \emph{weak topological support} contains all Borel probability measures on the membrane.

It is worth being precise about what the universal average does and does not explain, in relation to Gleason's theorem \cite{Gleason1957}. In the extended Bloch representation the Born rule results from two distinct ingredients: the \emph{orthogonal} fall of the point-particle onto the membrane, which is deterministic and fixes the linear dependence of the outcome probabilities on the state (the Born probabilities are the barycentric coordinates $r_i^\parallel$ of the orthogonal projection $\rpar$), and the \emph{uniform} disintegration of the membrane, which is indeterministic. The universal average concerns the second ingredient only: it shows that, given the geometry of the membrane and the landing point of the particle, indifference forces the disintegration to be uniform. It says nothing about the first ingredient, and it is precisely there that the content of Gleason's theorem lies. Indeed, as shown in \cite{Sassoli2026}, for $N=2$ a non-orthogonal landing rule produces outcome probabilities which are perfectly consistent and yet non-Born. The uniform breaking of the band would then yield these non-Born probabilities, and so would, by Theorem~\ref{thm:uniqueness}, every universal measurement, since the universal average acts on the way the band disintegrates, not on the way the particle lands. For $N\ge3$, on the contrary, the normalization constraints eliminate nonlinear rules (within the restricted class of scalar-product deformations studied in \cite{Sassoli2026}). This provides a Bloch-geometric illustration of Gleason-type restrictions. So, the two results are complementary and concern the two stages of the measurement process: within the particular scalar-product deformation examined in \cite{Sassoli2026}, Gleason-type normalization constraints restrict the first stage for $N\ge3$ (while for $N=2$ the orthogonal fall remains a genuine assumption of the model); the universal average fixes the indeterministic second stage, for every $N\ge2$. The heuristic argument of \cite{Aerts1998,Aerts1999}, which derived the universal transition probabilities directly from Gleason's theorem under the assumption that they depend only on the initial and final states, can be seen in retrospect as a conflation of the two stages. Degenerate measurements, which in \cite{AS2014a} require an additional deterministic process after the collapse, are treated in the same way, since this additional process does not involve the selection of the interaction.

Turning now to Bertrand's paradox, in \cite{AS2014b} the same cellular construction was applied to the different ``ways of tossing a stick'' onto a circle, described by densities $\rho(x,\theta)$ on the parameter space $[0,2]\times[0,\pi]$ of the lines meeting the circle ($\theta$ the direction, $x$ the distance of the line from the bottom of the circle), and the universal average was found to give the uniform density on this parameter space, hence the value $1/2$ for the probability that the chord is longer than the side of the inscribed triangle, which is also the value obtained by Jaynes \cite{Jaynes} from the requirement of invariance under rotations, translations and dilations of the circle. Similarly, for the ``ways of tossing two pebbles'' onto the circumference, the universal average gives $1/3$, and for a pebble tossed inside the disk, $1/4$.

The results of the present article apply to these cases without change: for each of these parameter spaces, with its area measure, the universal average is unique, exact at every resolution, given by every indifferent random choice of a density, and the two-valued cellular densities are inessential. In this sense, the questions (Q1)--(Q3), which were raised by Shackel precisely in the context of Bertrand's paradox, receive here their answer. The question (Q4), i.e., the dependence of the result on the parameterization, is in that context the crux of the matter, since Bertrand's question, as posed, specifies neither a randomized entity nor a geometry. Theorem~\ref{thm:length} tells us exactly what can and cannot be expected from a universal average in this respect: it cannot select a notion of size, but once one is fixed by the modelization of the problem, the universal average is unique. We believe that this allows a stronger conclusion than the one we drew in \cite{AS2014b}. Bertrand's question asks for a chord chosen ``at random,'' i.e., it appeals to indifference, and Theorem~\ref{thm:length} shows that indifference, however it is implemented -- by a direct application of the principle of insufficient reason to the chords, or by a universal average over all the ways of selecting them, or by any other procedure that does not privilege some chords over others -- cannot produce a notion of size out of nothing: it can only be indifference \emph{with respect to} a notion of size that has to be supplied independently. Consequently, no solution of Bertrand's paradox can avoid the separation between the two problems: an ``easy'' part, which specifies a physical realization of the randomization, thereby identifying the randomized entity and the natural geometric measure (with its associated symmetries) on the space of possibilities, and a ``hard'' part, which, given that specification, determines the outcome probabilities \cite{AS2014b}. The distinction between the two problems is thus forced upon us by Theorem~\ref{thm:length}, and is not a matter of interpretation. What the universal average adds is that, once the easy problem is solved, the hard one has a unique, exact and robust solution. The paradox arises, and will keep arising, whenever the easy problem is left implicit.

\section{Conclusion}\label{sec:conclusion}

Let us return to the question in the title: are quantum measurements universal measurements? The results of this article allow us to answer it in a more precise way than in \cite{AS2014a}.

First, we have identified what a universal measurement is: a random way of selecting a $\rho$-measurement, described by a probability law $\U$ on the set of all possible densities --- what we have called a $\U$-measurement --- and subject to the condition that makes it \emph{universal}: the law $\U$ must be \emph{indifferent}, i.e., it must not privilege any portion of the elastic band of the model, in the sense of being invariant under the cutting and re-gluing of the band in a different order. This is the transcription, at the level of the selection of the measurement, of the condition of maximal lack of knowledge which, at the level of the selection of the interaction, is expressed by the uniform density.

Second, we have shown that the outcome probabilities of \emph{every} universal measurement are the quantum probabilities. The universal average is unique, it is exact at every finite resolution, and it does not depend on the way the random selection is constructed: by finite families of cellular densities as in \cite{AS2014a}, by continuous families, or by random mixtures of deterministic measurements of the Dirichlet type, whose \emph{weak topological support} covers the entire space of probability measures on the membrane. In this precise sense, the quantum measurement is the universal measurement: it is the common ``center of mass'' of all the indifferent ways of selecting a measurement. The cellular construction of \cite{AS2014a} was one way of exhibiting this center of mass, and its apparent peculiarities (equal cells, two values, constant sequence of averages, combinatorial identities) are all explained by the underlying symmetry.

Third, we have shown that the quantum measurement is not only the average of every universal measurement, but the \emph{only} universal measurement free of deterministic components. The fluctuations of a universal measurement around its quantum average are entirely due to its deterministic components, and are governed by a single number, the coincidence probability $\beta$, which also controls
 the correlations between the outcomes of a sequence of measurements performed with the same apparatus. The Born rule, together with the independence of repeated measurements, is thus recovered either when the way of selecting the interaction is renewed at each run, or when it is persistent but free of deterministic components, and the notion of universal measurement provides, with $\beta$ and the persistence time of the way of selecting the interaction, the parameters that quantify possible deviations.

Fourth, we have shown that the only input of the whole construction which cannot be derived from the notion of universal average itself is the notion of length (or volume) on the membrane representing the measurement, i.e., the Euclidean geometry of the extended Bloch sphere, which is itself a consequence of the unitary invariance of quantum mechanics. Indifference is always indifference with respect to a notion of size. The universal average does not create the geometry of the measurement, it shows that, given this geometry, and a point-particle falling orthogonally onto the disintegrable membrane, nothing else is needed to obtain the Born rule. The orthogonal fall, which encodes the linear dependence of the probabilities on the state, is the part of the model to which Gleason's theorem pertains, and the two results are complementary: Gleason's theorem constrains the deterministic stage of the measurement, for $N\ge3$, and the universal average fixes its indeterministic stage, for every $N\ge2$ (Section~\ref{sec:general}).

The answer to the question in the title is therefore affirmative, in the following precise sense. If a quantum measurement is a situation of maximal lack of knowledge, not only about the interaction that is actualized, but also about the way it is selected, and if this lack of knowledge is understood as indifference with respect to the natural geometry of the measurement, then the Born rule is not merely compatible with the hypothesis: it is forced by it, whatever the details of the way the interaction is selected, and no other probability rule is possible. Conversely, every indifferent way of selecting an interaction, whether it is realized by a uniform membrane, by a random collection of deterministic breaking points, or by any mixture of the two, is operationally a quantum measurement. Whether quantum measurements \emph{are} universal measurements in this sense, in nature, is of course a question that mathematics cannot settle. But the hypothesis is now shown to be consistent, unique in its consequences, and, through the statistics of sequences of measurements, not entirely beyond the reach of future experiments.

\appendix

\section{Combinatorial identities}\label{app:comb}

We derive two combinatorial identities (Eqs.~(108) and (113) of \cite{AS2014a}). We observe that under $\U_n$ the disintegrable cells form a random non-empty subset $S$ of $\{1,\dots,n\}$, all the $2^n-1$ non-empty subsets being equally likely, and $|S|=n^d$ is the number of disintegrable cells. We write $\eta_j(S)=1$ if $j\in S$ (the $j$-th cell is disintegrable) and $\eta_j(S)=0$ otherwise. Since the two-valued density $\rho_n$ associated with $S$ takes the value $n/2n^d$ in the disintegrable cells, of length $2/n$, and the value $0$ in the others, the mass of the $j$-th cell is
\begin{equation}\label{eq:mass-S}
m_j(\rho_n)=\frac{\eta_j(S)}{|S|},\quad j=1,\dots,n
\end{equation}
i.e., $1/n^d$ if the cell is disintegrable and $0$ if it is not. We can compute the mean value of the mass of the first cell in two different ways. The first way is a direct count: the subsets $S$ of size $k$ which contain the cell $1$ are obtained by choosing the other $k-1$ elements among the remaining $n-1$ cells, so there are $\binom{n-1}{k-1}$ of them, each contributing $1/k$ to the mass of the first cell. Hence
\begin{equation}\label{eq:direct}
\avg{m_1(\rho_n)}=\frac{1}{2^n-1}\sum_{k=1}^{n}\binom{n-1}{k-1}\frac1k
\end{equation}
The second way uses the symmetry. Summing \eqref{eq:mass-S} over $j$ gives $\sum_jm_j(\rho_n)=\sum_j\eta_j(S)/|S|=|S|/|S|=1$ for every $S$, so that the mean values of the $n$ cell masses sum to $1$, and since the uniform law on the subsets does not distinguish between the cells (a rearrangement of the band maps subsets of size $k$ onto subsets of size $k$), these $n$ mean values are all equal. Therefore
\begin{equation}\label{eq:symmetric}
\avg{m_1(\rho_n)}=\frac1n
\end{equation}
Comparing \eqref{eq:direct} and \eqref{eq:symmetric} we obtain the identity
\begin{equation}\label{eq:identity}
\sum_{k=1}^{n}\binom{n-1}{k-1}\frac1k=\frac{2^n-1}{n}
\end{equation}
which is nothing but the identity (113) of \cite{AS2014a}, written for $n-1$ instead of $n$. One can also check \eqref{eq:identity} directly, using $\binom{n-1}{k-1}\frac1k=\binom nk\frac1n$ and $\sum_{k=1}^n\binom nk=2^n-1$, but the point here is that it is the symmetry which makes the identity true, and which makes it hold for every $n$. Identity (108) of \cite{AS2014a}, on the other hand, can be obtained using (\ref{eq:identity}), with $n$ replaced by $n+1$, and the subtraction:
\begin{equation}\label{eq:identity108}
\sum_{k=0}^{n}\frac{k}{k+1}\binom nk=\sum_{k=0}^{n}\binom nk-\sum_{k=0}^{n}\frac{1}{k+1}\binom nk=2^n-\frac{2^{n+1}-1}{n+1}=\frac{2^n(n-1)+1}{n+1}
\end{equation}

\section{Approximation of an arbitrary density}\label{app:approx}

We show that, for every density $\rho$ on $[-1,1]$, possibly with delta components, there is a \emph{single} sequence of two-valued cellular densities $\rho_n$, in the sense of \eqref{eq:cellular}, such that $\PP(x\in A|\rho_n)\to\PP(x\in A|\rho)$ for every interval $A$ whose endpoints are not delta components of $\rho$, and uniformly over all intervals if $\rho$ has no delta components. 

This is the statement at which Section~13.1 of \cite{AS2014a} aimed. The argument given there, however, chose the number of disintegrable elementary cells separately in each cell, without controlling the total error accumulated over all the cells, so that, strictly speaking, it only established the weaker statement in which the approximating sequence is allowed to depend on the interval $A$ considered (this is the ambiguity in the order of the quantifiers pointed out in \cite{Shackel2024}). The construction below provides a single sequence, with an explicit bound on the error, valid for all intervals at once.

Let $m$ be the number of \emph{cells} $S_1,\dots,S_m$, of length $\frac2m$, and let each cell be divided into $\ell$ elementary cells, so that $n=m\ell$. Let $p_i=\int_{S_i}\rho$ be the cell masses of $\rho$, $\sum_ip_i=1$. Choose $\ell=\ell(m)=m^2+m$, and let the two-valued cellular density $\rho_n$ have, in the cell $S_i$, exactly $n_i=\lfloor\ell p_i\rfloor$ disintegrable elementary cells (the first $n_i$ ones, say), so that the total number of disintegrable elementary cells is $n^d=\sum_in_i$, with $\ell-m<n^d\le\ell$ (because each $\lfloor\ell p_i\rfloor$ differs from $\ell p_i$ by less than $1$); in particular $n^d\ge1$. The mass of $\rho_n$ in $S_i$ is $n_i/n^d$, and
\begin{equation}
\sum_{i=1}^m\Big|\frac{n_i}{n^d}-p_i\Big|\le\sum_{i=1}^m\frac{|n_i-\ell p_i|}{n^d}+\sum_{i=1}^m\ell p_i\Big(\frac1{n^d}-\frac1\ell\Big)<\frac{m}{n^d}+\frac{\ell-n^d}{n^d}<\frac{2m}{\ell-m}=\frac2m
\end{equation}
So, the cell masses of $\rho_n$ and of $\rho$ differ, in total, by less than $\frac2m$. For an interval $A$ whose endpoints fall in the cells $S_a$ and $S_b$, let $U$ be the union of the cells entirely contained in $A$: then $|\PP(x\in U|\rho_n)-\PP(x\in U|\rho)|<\frac2m$, while the parts of $A$ outside $U$ lie in $S_a\cup S_b$ and carry a mass at most $p_a+p_b$ under $\rho$ and at most $p_a+p_b+\frac2m$ under $\rho_n$. Altogether, we have
\begin{equation}
\big|\PP(x\in A|\rho_n)-\PP(x\in A|\rho)\big|<\frac4m+2(p_a+p_b)\le\frac4m+4\max_ip_i
\end{equation}
If $\rho$ has no delta components, $\max_ip_i\to0$ as $m\to\infty$, and the right-hand side tends to zero uniformly in $A$. If $\rho$ has delta components, $p_a+p_b$ tends to zero as long as the endpoints of $A$ are not delta components. In both cases a single sequence, $n=m\ell(m)$, $m=1,2,\dots$, works for all intervals.

\section{Concentration of the cellular measurements $\U_n$}\label{app:concentration}

We prove the bound \eqref{eq:concentration}. Under $\U_n$, the set $S$ of disintegrable cells is a random non-empty subset of $\{1,\dots,n\}$, all the $2^n-1$ non-empty subsets being equally likely, and for a region $A$ made of $a$ whole cells the probability of the outcome associated with $A$ is the fraction of disintegrable cells that lie in $A$:
\begin{equation}\label{eq:ratio}
\PP(x\in A|\rho_n)=\frac{|S\cap A|}{|S|}
\end{equation}
where $|S|=n^d$ is the number of disintegrable cells and $|S\cap A|$ the number of those lying in $A$. Let us first fix the number $|S|=k$ of disintegrable cells. Since all subsets are equally likely, the subsets of size $k$ are equally likely among themselves, i.e., $S$ is obtained by drawing $k$ cells at random, without replacement, among the $n$ cells, of which $a$ belong to $A$. The number $|S\cap A|$ of drawn cells belonging to $A$ is then a classical \emph{hypergeometric} random variable, whose mean and variance are
\begin{equation}
\avg{|S\cap A|}_k=k\,\frac an,\quad \Var_k\big(|S\cap A|\big)=k\,\frac an\Big(1-\frac an\Big)\frac{n-k}{n-1}\le\frac k4
\end{equation}
where the subscript $k$ recalls that $|S|=k$ is fixed, and we used $\frac an(1-\frac an)\le\frac14$, and $\frac{n-k}{n-1}\le1$. Dividing by $k$, the ratio \eqref{eq:ratio} has, for fixed $|S|=k$, mean value $\frac an$ (whatever $k$: this is Theorem~\ref{thm:finite} once more) and variance at most $1/4k$. Now, the variance of a random variable $X$ whose law depends on an auxiliary random quantity (here $k=|S|$) is given by the law of total variance: it is the mean, over $k$, of the conditional variance of $X$ for fixed $k$, plus the variance, over $k$, of the conditional mean of $X$ for fixed $k$. Here the conditional mean is $a/n$ for every $k$, so its variance over $k$ vanishes, and we are left with
\begin{equation}\label{eq:totvar}
\Var_{\U_n}\big[\PP(x\in A|\rho_n)\big]=\Big\langle\Var_k\Big(\frac{|S\cap A|}{k}\Big)\Big\rangle_{\!\U_n}\le\frac14\,\Big\langle\frac{1}{|S|}\Big\rangle_{\!\U_n}
\end{equation}
Since there are $\binom nk$ subsets of size $k$, each with probability $\frac{1}{2^n-1}$,
\begin{equation}
\Big\langle\frac{1}{|S|}\Big\rangle_{\!\U_n}=\frac{1}{2^n-1}\sum_{k=1}^n\binom nk\frac1k\le\frac{2}{2^n-1}\sum_{k=0}^n\binom nk\frac{1}{k+1}=\frac{2}{2^n-1}\frac{2^{n+1}-1}{n+1}\le\frac{6}{n+1}
\end{equation}
where we used $\frac1k\le\frac{2}{k+1}$ for $k\ge1$, then the identity \eqref{eq:identity} with $n$ replaced by $n+1$, and finally $\frac{2^{n+1}-1}{2^n-1}\le3$ for $n\ge1$. Inserting this into \eqref{eq:totvar} gives $\Var_{\U_n}[\PP(x\in A|\rho_n)]\le\frac{3}{2(n+1)}\le\frac{2}{n+1}$, which is \eqref{eq:concentration}. The bound is not sharp: a numerical evaluation gives $n\Var_{\U_n}[\PP(x\in A|\rho_n)]\approx\frac14$ for $n=10,10^2,10^3$, and $A$ half of the band, but what matters is that the variance tends to zero as $n\to\infty$, so that the outcome probabilities of the randomly selected two-valued density concentrate on the quantum values.

\section*{Acknowledgments} I would like to express my gratitude to Diederik Aerts for the wonderful scientific collaboration we have shared over the past fifteen years, and which continues to this day. After a long period away from academic research, my interest was rekindled by my encounter with him and our initial discussions of the problem at the heart of the present article: his conjecture on universal measurements, which eventually led us in 2014 to the result presented in \cite{AS2014a}. I am also grateful to N. Shackel, whose criticisms, though severe, were undoubtedly pertinent and provided me with the motivation to revisit this fascinating notion, which I believe has now found a more mature formulation in the present work.  

\section*{Declaration of generative AI use} The mathematical reformulation of the universal average as an invariant average presented in this article, including the statements and proofs of Theorems 1 through 5, the connection to Dirichlet processes, and a first draft of the manuscript, were developed by the author drawing also on a dialogue with the artificial intelligence assistant Claude (Anthropic).

\end{document}